\documentclass[11pt]{article}
\usepackage{amsmath,amsfonts,amsthm,amssymb,ifthen,graphicx}
\usepackage{bbm}
\usepackage{url}
\graphicspath{{./}{./figures/}}
\newboolean{@laptop}
\newboolean{@todoon}
\setboolean{@laptop}{true}
\setboolean{@todoon}{true}

\ifthenelse{\boolean{@laptop}}{}{\usepackage{mathbbol}}

\def\todo#1{\ifthenelse{\boolean{@todoon}}{\marginpar{\textit{#1}}}{}}

\usepackage{times}
\usepackage{tabu}
\usepackage{thmtools}
\usepackage{thm-restate}

\usepackage{helvet}
\usepackage{courier}
\usepackage{colortbl,xcolor}
\usepackage{amsmath}
\usepackage{amssymb}

\usepackage[noend]{algpseudocode}

\usepackage{mathabx}
\usepackage{algorithm}
\usepackage{algorithmicx}
\usepackage[noend]{algpseudocode}
\newtheorem{lemma}{Lemma}

\newtheorem{theorem}[lemma]{Theorem}

\newtheorem{remark}{Remark}

\newtheorem{corollary}[lemma]{Corollary}

\DeclareMathOperator*{\argmax}{\mathop{\arg\max}}

\newcommand{\preal}{\mathbb{R}_+}
\newcommand{\nreal}{\mathbb{R}_-}
\newcommand{\real}{\mathbb{R}}
\newcommand{\dprob}{Di-LoMuF}

\newcommand{\umaxfeprob}{maxf-\uprob}
\newcommand{\uprob}{LoMuF}
\newcommand{\prob}{MCF}
\newcommand{\dem}{\vec{d}}
\newcommand{\supp}{\vec{s}}
\newcommand{\capa}{\vec{c}}
\newcommand{\flow}{\vec{f}}

\begin{document}

\title{Target Location Problem for Multi-commodity Flow}

\author{
Xingwu Liu \thanks{
Institute of Computing Technology,
Chinese Academy of Sciences. University of Chinese Academy of Sciences. Beijing, China.
Email:\texttt{liuxingwu@ict.ac.cn}. 
}
\and
Zhida Pan \footnote{Corresponding author} \thanks{
Institute of Computing Technology,
Chinese Academy of Sciences. University of Chinese Academy of Sciences. Beijing, China.
Email:\texttt{zhidapan@gmail.com}. 
}
\and
Yuyi Wang \thanks{
ETH Zurich, Switzerland.
Email:\texttt{yuwang@ethz.ch}. 
}
}

\date{April 15, 2020}

\maketitle
\begin{abstract}
Motivated by scheduling in Geo-distributed data analysis, we propose a target location problem for multi-commodity flow (\uprob~for short). Given commodities to be sent from their resources, \uprob~aims at locating their targets so that the multi-commodity flow is optimized in some sense. \uprob~is a combination of two fundamental problems, namely, the facility location problem and the network flow problem. We study the hardness and algorithmic issues of the problem in various settings. The findings lie in three aspects. First, a series of NP-hardness and APX-hardness results are obtained, uncovering the inherent difficulty in solving this problem. Second, we propose an approximation algorithm for general undirected networks and an exact algorithm for undirected trees, which naturally induce efficient approximation algorithms on directed networks. Third, we observe separations between directed networks and undirected ones, indicating that imposing direction on edges makes the problem strictly harder. These results show the richness of the problem and pave the way to further studies.
\end{abstract}

\newpage

\section{Introduction}
Nowadays, data is generated geo-distributively at a much higher speed as compared to the existing data transfer speed; for instance, telescopes around the world bring us an unimaginable amount of astronomy data. There are two main reasons for having geo-distributed data: (1) Datacenters (DCs) are built across the globe. (2) Organizations prefer to use multiple clouds to increase reliability, security, and processing. Besides, there exist applications that process and analyze a huge amount of massively geo-distributed data to extract useful information. A typical scenario in processing geo-distributed data is that several analysis tasks are running simultaneously, and each requires a fraction of the collected data \cite{NagaA2018,PuQifan2015,Yue2016Fast,Yin2015,Khuller2016}. In addition, every analysis task moves needed data to a single location before the computation. Fig.\ \ref{fig:astronomy} shows an example of geo-distributed telescope data. 

The network bandwidth is a crucial factor in geo-distributed data movement and becomes the resource bottleneck. For example, the demand for bandwidth increased from 60 to 290 Tbps between the years 2011 and 2015 while the network capacity growth was not proportional. In 2015, the network capacity growth was only 40 percent, which was the lowest during the years 2011 and 2014\footnote{\url{https://www.telegeography.com/researchservices/global-bandwidth-research-service/}}. 
When applications (such as electromagnetic radiation and infrared ray analysis), each handling data from some datacenters, have to be deployed, there is no meaningful notion of distance. The latency (travel time of a single small packet) under low-congestion conditions tends not to be noticeable to the end-users. The real difficulty here is the underlying capacity of the network. If links become congested, then the latency will increase and throughput will suffer. A key issue is how to allocate enough bandwidth to each application without causing congestion on the network \cite{Zhao2020}. 

Hence, we need to choose proper locations for the tasks to reduce congestion. Specifically, we propose this target location problem for multi-commodity flow: Given sources of multiple commodities on a capacitated network, the goal is to locate the targets to maximize the flow value. 

We fix sources because, as our motivating example of geo-distributed data analysis shows, it is difficult, if not impossible, to change datacenters that collect data since for efficiency as these datacenters should be close to data generators. However, it is much more flexible to choose the target locations where the analysis tasks are performed. 

The multi-commodity flow problem (MCF) is one of the most fundamental problems with a wide variety of scientific and engineering applications that have been studied intensively \cite{Shepherd2015,Monis2019}. In the most typical scenario, a finite number of commodities have to be sent from their sources to targets on a capacitated network. Each commodity has its own flow, and the commodities interact when their flows compete for capacity on common edges. 

There are two general classes of MCF. One is network analysis which, based on a given network configuration, finds the optimal flow pattern for some objective function. The most studied objective functions include maximizing flow values and minimizing flow costs. The other belongs to network synthesis which seeks an optimal network configuration satisfying certain requirements. 

In both classes, the targets of the commodities are taken for granted. To our surprise, researches have long neglected how targets are chosen. This paper is devoted to initializing such a theory. 

\begin{figure}
    \centering
    \includegraphics[scale = 0.4]{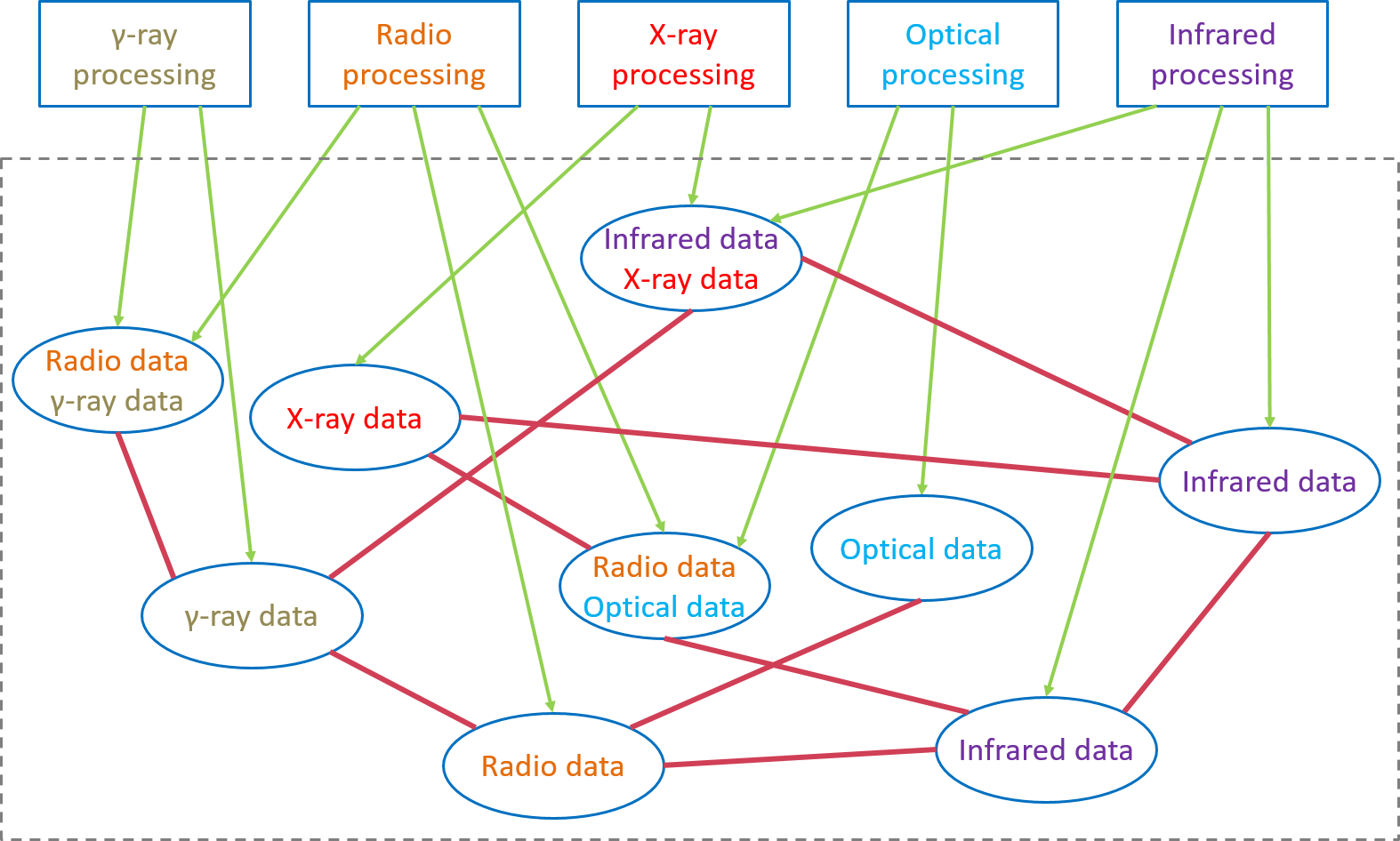}
    \caption{An overview of geo-distributed astronomy data and corresponding tasks}
    \label{fig:astronomy}
\end{figure}

Our proposed problem extends the MCF framework. It does not belong to either of the two classes. It is a combination of the facility location problem and the network flow problem which are inherently related \cite{An2014}. Facility location is a branch of operations research related to locating or positioning at least a new facility among several existing facilities to optimize (minimize or maximize) at least one objective function. It is among the most fundamental problems in operations research and theoretical computer science \cite{Vazirani2001Approximation}. Facility location met network flow in 1990 \cite{Tamura1990} and has inspired a series of work \cite{ Ito2009The,ArataLocating2002,Kortsarz2008,Andreev2009Simultaneous,Hebler2016}. However, all the published works minimize costs of the selected sources or targets (not the flow cost which, together with flow value, is the objective of network flow problems), and never consider the multi-commodity setting. The most crucial difference lies in that our combination is inherent, meaning that the objective is to optimize the flow value, but the literature focuses on the cost of the selected nodes rather than the cost of flow. Another benefit of our framework is that it can naturally extend almost every network flow problem, e.g., flow cost minimization. 

Our model has other applications, such as Web server deployment. There are serving various demands from widely-distributed users, for example, requesting for different online video, where should the servers be located so that the users have a good experience? Again we do not care distance, and the key objective is to optimize the available bandwidth. There are more motivating examples, say, network-flow based evacuation planning for an emergency where shelters have to be selected, and congestion decides the efficiency of evacuation. Interested readers are referred to \cite{Hebler2016}.

These real-world scenarios justify our problem's critical features: The commodities are only partially determined since the targets are not given and have themselves to be optimized, and a decisive factor of the optimization is the bandwidth rather than any notion of distance. This well motivates our problem.






\subsection{Results and Discussion}
We propose a novel model of the target location for multi-commodity flow (\uprob). On the one hand, we figure out the hardness results of various versions. On the other hand, we design algorithms for several versions. The results are as follows. 
\begin{enumerate}
\item We show that the \uprob~problem is NP-hard on general undirected graphs. 

We know that if the targets are fixed, the problem degenerates to one normal multi-commodity flow problem (allowing fractional flows) and becomes tractable in polynomial time, which shows that the most challenging part of this problem is indeed how to locate the targets. 
\item We design a polynomial-time algorithm solving \uprob~on trees. 

Trees are important network structures in practice. Compared to the NP-hardness result above, the fact that there is only one path connecting a source and the target on a tree simplifies the problem. Our algorithm is elegant and surprisingly shows that the interaction between different commodities becomes not harmful on trees.  

\item We present a $\max\{\theta-1,1\}$-approximation algorithm for \uprob~on general undirected graphs, where $\theta$ is the largest source number among all commodities.

This result actually shows that, when $\theta \le 2$ (the so-called bi-source cases) the problem can also be solved efficiently, but (take into account the NP-hardness result above) becomes intractable when $\theta\ge 3$. 

\item For \uprob~on directed graphs (\dprob), we prove that it is also NP-hard and even cannot be efficiently approximated with a ratio less than 2. 

In fact, we also show that \uprob~on undirected graphs can be reduced to the directed case \dprob, and then \dprob~should be even harder. 

\item \dprob~also remains NP-hard on symmetric di-paths and bi-source supply vectors. 

These are clear separations between undirected \uprob~and \dprob, since undirected \uprob~is efficiently solvable on trees while \dprob~is even difficult on paths, a very special case of trees. As we pointed out above, the bi-source instances are easy for undirected \uprob, but not for \dprob. 

\item For the special case on symmetric di-trees, \dprob~has a polynomial-time 2-approximation algorithm.

Though we have seen several hardness results of \dprob, for a special but still meaningful subset, where every link has the same capability of downloading and uploading, we can obtain an efficient approximation algorithm. 

\item We show that our results above can also be extended to other variants of \uprob~such as maximum sum flows, unsplittable flows, restricted candidate targets, maximum feasible flows, and so on. For the unsplittable version, we show that it cannot be approximated within ratio 2. For the version with restrictions on targets, it is NP-hard on uni-source supply vectors and stars and cannot be efficiently approximated within ratio $\frac{7}{6}$ on trees. For the maximum feasible flows version, we prove that for any constant $\epsilon>0$, unless NP=ZPP, it cannot be approximated within $O(k^{1-\epsilon})$ on $k$ supply vectors. 

This shows that the framework of the new location problems has a powerful capability of modeling different scenarios in practice and enriches the theory of location problems and network flow problems. 

\end{enumerate}

\subsection{Related Work}
There is an increasing vast literature on multi-commodity flow and its single-commodity special case \cite{Shepherd2015}. Basically, there are two types of optimization objectives, namely, minimum cost and maximum flow which is the focus of this paper. The main theme of maximum flow in recent years is improving the efficiency of approximation algorithms \cite{Monis2019,Kelner2014,Teng2010,LeeYinTat2013,PengRichard2016,Madry2016,Kelner2012,Sherman2017,Chekuri2013The}. The flow-cut duality is also a challenging issue and has attracted much attention from researchers \cite{krauthgamer2019flow,salmasi2019constant}. 

Facility location has flourished ever since the 1960s and remains an active topic in operations research and theoretical computer science \cite{Hakimi1964,Labbe1992,Vazirani2001Approximation,Labbe1992}. Though generally, no constant-ratio approximation algorithm exists, it can be constant-approximated on metric spaces. One of the main threads of research is to improve the approximation ratio in various situations \cite{Shmoys1997Approximation,LiA,GuhaGreedy}. 

Though inherently related to multi-commodity flow\cite{An2014}, facility location got to be combined with network flow only in 1990 \cite{Tamura1990}, when the source location problem was proposed. Roughly speaking, the mission of the source location problem on a network is to find a set of sources from which enough flow can be sent to each prescribed target. In addition to flow requirements, connectivity and vertex coverage are also frequently used constraints. Work in this line can be classified into two categories. One is independent source location, meaning that the flows to different targets do not interact \cite{ArataLocating2002,SakashitaMinimum,Attila2008,MamadaOptimal2002,MamadaAn2006,Ito2009The,ArataLocating2002,Kortsarz2008}. The other is simultaneous source location, where the flows concurrently exist and interact by competing edge capacities \cite{Andreev2009Simultaneous,Hebler2016}. An interesting application is emergency evacuation planning \cite{Hebler2016}, where shelters are to be located where residents in a disaster can move to as fast as possible. In such applications, capacities are also usually imposed on network nodes, rather than just on edges in typical network flow models. All the mentioned works have two common features. First, essentially only a single commodity is considered which is multi-source multi-target. Second, the objective is to optimize some measures of the selected sources (say, total cost), rather than the properties of the flow (say, flow value). This is in sharp contrast to our proposed problem.

\section{Preliminaries and Problem Statement}\label{sec:undigraphs}
In this section, we review key notions and notations used in this paper, and formally define the location problem. 

\subsection{Preliminaries}\label{Preliminaries}
Let $\real$ ($\preal,\nreal$, respectively) represent the set of (non-negative, non-positive, respectively) real numbers. We use $\vec{x}$ for a vector, and $\vec{x}(y)$ for its $y$-th entry. When we denote a set by an upper-case letter, we usually write the corresponding (subscripted) lower-case letter for the members.  

A network is a capacitated graph $G=(V,E,\capa)$, where $V$ is the vertex set, $E$ is the edge set, and $\capa\in \preal^E$ assigns capacities to the edges. 
We first mainly focus on undirected graphs in this paper, and will consider directed graphs in Section \ref{sec:directedgraphs}. 
For any $v,v'\in V$, we use $\langle v,v'\rangle$, or $\langle v',v\rangle$ interchangeably, to denote the edge between $v,v'$. A commodity is described by a demand vector $\dem\in \real^V$ satisfying $\sum_{v\in V}\dem(v)=0$, where any $v$ such that $\dem(v)<0$ ($\dem(v)>0$, respectively) is called a source (a target, respectively). Intuitively, each source $v$ has to sent out $\dem(v)$ units of the commodity, and in total $\dem(u)$ units are delivered to target $u$. 
The vertex set of a graph $G$ is denoted by $V(G)$. 

To specify flows over a network, we always arbitrarily orient all the edges and keep the orientation implicit unless necessary. For any $v\in V$, let $E_-(v)$ ($E_+(v)$, respectively) stands for the set of incoming (outgoing, respectively) edges. A flow is a vector $\flow\in \real^E$, which for any edge $e\in E$, means $|\flow(e)|$ units of transportation along $e$ in orientation if $\flow(e)>0$, and opposite direction otherwise. Given flows $\flow,\flow'\in \real^E$, we write $\flow\lesssim\flow'$ if $|\flow|\le |\flow'|$ for any $e\in E$. A flow $\flow$ is said to satisfy a demand vector $\dem$, if for any $v\in V$, $\dem(v)=\sum_{e\in E_-(v)} \flow(e) - \sum_{e\in E_+(v)} \flow(e)$. A multi-commodity flow, which means a set $F$ of flows, is valid if its congestion $\sum_{\flow\in F}|\flow(e)|$ along any edge $e\in E$ is at most $\capa(e)$. 

The maximum concurrent problem (\prob~for short) has been extensively and is still being actively studied. Specifically, given demand vectors $\dem_i,1\le i\le k$ on a capacitated  graph $G$, the mission of \prob~is to find the maximum $\lambda$ such that $\lambda\dem_i,1\le i\le k,$ can be satisfied by a valid multi-commodity flow on $G$. The optimum $\lambda$ will be denoted by $\lambda(G;\dem_1,\cdots,\dem_k)$. 

Let's recall some properties of \prob. 
\begin{lemma}\label{le:MuFinP}
\prob~lies in P.
\end{lemma}

As mentioned in \cite[page~863]{cormen2009introduction}, there is no known purely combinatorial algorithm solving \prob~exactly and efficiently. The only commonly used algorithm is based on linear programming. 

A multi-commodity flow $F$ on a capacitated graph $G$ is said to be a decomposition of flow $\flow$, if $\flow(e)=\sum_{\flow'\in F}\flow'(e)$ and $|\flow(e)|=\sum_{\flow'\in F}|\flow'(e)|$ for any edge $e$ of $G$. 
\begin{lemma}\label{le:decomposable}
Arbitrarily fix a demand vector $\dem$ on a capacitated graph $G$. Suppose $\dem$ has exactly one target $t$. Then any flow satisfying $\dem$ can be decomposed into a multi-commodity flow which satisfies the demand vectors $\{\dem_v: v\textrm{ is a source of }\dem\}$. Here, each $\dem_v$ is such that for any vertex $u$ of $G$,  
$$\dem_v(u)=\begin{cases}
                       \dem(v)& \textrm{if }u=v\\
                       -\dem(v)& \textrm{if }u=t\\
                       0& \textrm{otherwise}
                  \end{cases}.$$
\end{lemma}
Note that the decomposition in Lemma \ref{le:decomposable} is not necessarily unique. Any such one will be called a canonical decomposition of the flow.

Given any vertex subset $U$ of an graph $G=(V,E,\capa)$, the cut induced by $U$, denoted by $Cut(U)$, is defined to be the set of edges bridging $U$ and $V\setminus U$. Let $E_-(U)=Cut(U)\bigcap (\bigcup_{u\in U} E_-(u))$ be the set of edges coming into $U$, and $E_+(U)=Cut(U)\setminus E_-(U)$. 
\begin{lemma}\label{le:cutflow}
Suppose that $\flow$ is a flow satisfying a demand vector $\dem$ on a capacitated graph $G=(V,E,\capa)$. Then for any $U\subseteq V$, $\sum_{e\in E_-(U)} \flow(e) - \sum_{e\in E_+(U)} \flow(e)=\sum_{u\in U}\dem(u)$.
\end{lemma}
\subsection{Target Location problem}
Intuitively, our goal is to properly locate targets for multiple commodities. We formulate this problem in this subsection. 

Given a capacitated graph $G=(V,E,\capa)$, any $\supp\in \nreal^V$ is called a supply vector on $G$. For any supply vector $\supp$ and $v\in V$, we define a demand vector $\supp \circ v$ such that for any $u\in V$, 
$$(\supp \circ v)(u)=\begin{cases}
                       \supp(u)& \textrm{if }u\neq v\\
                       -\sum_{w\in V\setminus\{v\}}\supp(w)& \textrm{otherwise}
                  \end{cases}.$$

It is time to formulate the problem of target location for maximizing concurrent multi-commodity flow, \emph{\uprob}~for short. Given supply vectors $\supp_1,\cdots,\supp_k$ on a capacitated graph $G$,  \uprob~aims at finding $v_1,\cdots,v_k$ such that $\lambda(G;\supp_1\circ v_1,\cdots,\supp_k\circ v_k)$ is maximized. By abuse of the notation, the optimum objective value is again denoted by $\lambda(G;\supp_1,\cdots,\supp_k)$. 



\section{Hardness and Algorithms of \uprob}
We begin with studying the hardness of \uprob. Our work refers to a well-known NP-complete problem, 3-dimensional matching (3-DM for short). Though \uprob~is NP-hard in general, we devise an algorithm solving \uprob~problems on trees efficiently, and show that a simple strategy could be a not-bad solution for graphs with bounded sources. 

\subsection{Hardness Result}
A 3-DM instance is a quadruple $(X,Y,Z,W)$, where $X,Y,Z$ are pairwise disjoint finite sets of equal size, and $W\subseteq \{\{x,y,z\}: x\in X,y\in Y,z\in Z\}$. The goal is to decide whether $W$ contains a perfect matching, namely, a subset $W'\subseteq W$ such that $|W'|=|X|$ and $\bigcup_{w\in W'}w=X\bigcup Y\bigcup Z$? The trivial cases where $\bigcup_{w\in W}w\neq X\bigcup Y\bigcup Z$ will not be considered.

We first show that \uprob~is NP-hard, which is more or less a surprise, compared with Lemma \ref{le:MuFinP}.
\begin{theorem}\label{thm:UndiNPC}
Given supply vectors $\supp_1,\cdots,\supp_k$ on a capacitated graph $G$, it is NP-complete to decide whether $\lambda(G;\supp_1,\cdots,\supp_k)\ge 1$.
\end{theorem}
\begin{proof}
Choose a target $v_i$ for supply vectors $\supp_i$, for any $1\le i\le k$. Due to Lemma \ref{le:MuFinP}, we can use $v_1,\cdots,v_k$ as a certificate to check whether $\lambda(G;\supp_1,\cdots,\supp_k)\ge 1$. This means that the decision problem lies in NP.

To prove NP-completeness, it suffices to establish a reduction from 3-DM.

\begin{figure}
  \centering
  \includegraphics[width=0.65\textwidth]{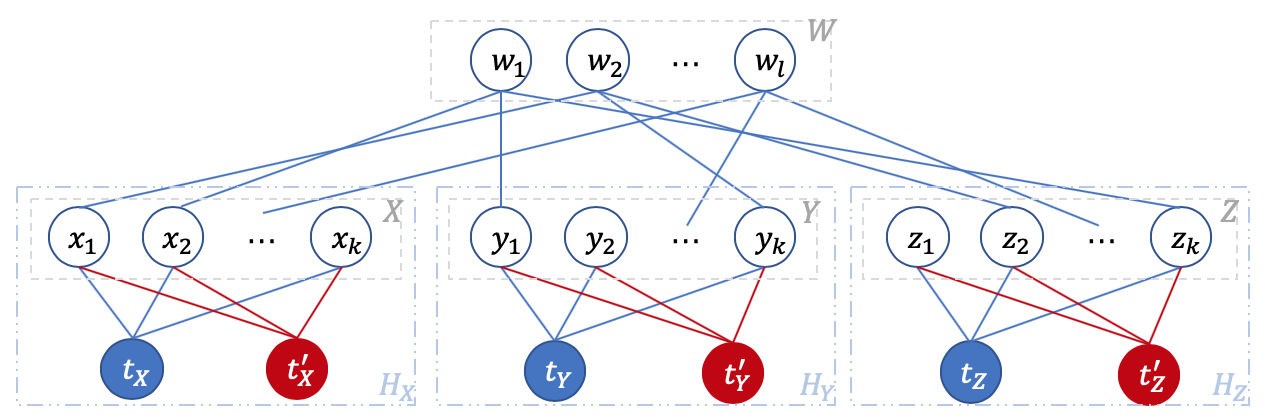}
\caption{The graph to which 3-DM is reduced.}\label{fig:UndiNPC}
\end{figure}   

Given a 3-DM instance $(X,Y,Z,W)$ with $|X|=k$ and $|W|=l$, we construct an capacitated graph $G=(V,E,\capa)$ as illustrated in Figure \ref{fig:UndiNPC}. Specifically, $G$ consists of three subgraphs $H_X,H_Y,H_Z$ connected via $W$. $H_X$ is a complete bipartite graph of vertex sets $X$ and $T_X=\{t_X,t'_X\}$, and any $x\in X$ is adjacent to $w\in W$ if and only if $x\in w$, likewise for $H_Y,H_Z$. All the edges are oriented upward in Figure \ref{fig:UndiNPC}.

As to the capacity, let $E'$ be the set of red edges, namely, those incident to $t'_X,t'_Y$ or $t'_Z$. For any $e=\langle v,t \rangle\in E'$ with $v\in X\bigcup Y\bigcup Z$, let $W_e=\{w\in W: v\in w\}$. Then for any $e\in E$,
$$\capa(e)=\begin{cases}
                       |W_e|-1& \textrm{if }e=E' \\
                       1& \textrm{otherwise}
                  \end{cases}.$$
We define $l$ supply vectors $\dem_1=\cdots=\dem_k,\dem_{k+1}=\cdots=\dem_l$ such that for any $v\in V$,
$$\dem_1(v)=\begin{cases}
                       -1& \textrm{if }v\in \{t_X,t_Y,t_Z\}\\
                       0& \textrm{otherwise}
                  \end{cases}$$
$$\dem_{k+1}(v)=\begin{cases}
                       -1& \textrm{if }v\in \{t'_X,t'_Y,t'_Z\}\\
                       0& \textrm{otherwise}
                  \end{cases}$$                  

The rest of the proof is devoted to showing that $W$ has a perfect matching if and only if the \uprob~instance satisfies $\lambda(G;\dem_1,\cdots,\dem_l)\ge 1$, which will lead to NP-completeness of our decision problem. The proof consists of two parts.

\textbf{Part 1}: a perfect matching in $W$ implies $\lambda(G;\dem_1,\cdots,\dem_l)\ge 1$.

Without loss of generality, suppose $\{w_1,\cdots,w_k\}\subset W$ is a perfect matching. For any $1\le i\le k$, define flow $\flow_i$ such that for any edge $e\in E$,
$$\flow_i(e)=\begin{cases}
                       1& \textrm{if } e \textrm{ is incident to }w_i \textrm{, or } e=\langle t,u\rangle  \textrm{ with }t\in\{t_X,t_Y,t_Z\}\textrm{ and } u\in w_i\\
                       0& \textrm{otherwise}
                  \end{cases}.$$   
For any $k+1\le j\le l$, define flow $\flow_j$ such that for any edge $e\in E$,
$$\flow_j(e)=\begin{cases}
                       1& \textrm{if } e \textrm{ is incident to }w_j \textrm{, or } e=\langle t,u\rangle  \textrm{ with }t\in\{t'_X,t'_Y,t'_Z\}\textrm{ and } u\in w_j\\
                       0& \textrm{otherwise}
                  \end{cases}.$$   

It is straightforward to check that the multi-commodity flow $\flow_1,\cdots,\flow_l$ is valid and satisfies the demand vectors $\dem_1\circ w_1,\cdots,\dem_l\circ w_l$. Hence, $\lambda(G;\dem_1,\cdots,\dem_l)\ge 1$.               

\textbf{Part 2}:  $\lambda(G;\dem_1,\cdots,\dem_l)\ge 1$ implies a perfect matching in $W$.

Suppose the optimum targets are $v_1,\cdots,v_l$, and the multi-commodity flow $\flow_1,\cdots,\flow_l$ is valid and satisfies $\dem_1\circ v_1,\cdots,\dem_l\circ v_l$. Part 2 immediately follows from the two facts:

               
\begin{enumerate}
\item Fact 1: $v_i\in W$ for any $1\le i\le l$. 

Consider the congestion of any $\flow_i$ on $Cut(W)=Cut(V(H_X))\bigcup Cut(V(H_Y))\bigcup Cut(V(H_Z))$. Let's proceed case by case.
\begin{itemize}
\item $v_i\in W$. Applying Lemma \ref{le:cutflow} to $\flow_i,\dem_i\circ v_i$, we see that the congestion of $\flow_i$ on $Cut(W)$ is at least 3.
\item $v_i\notin W$. Without loss of generality, assume $v_i\in V(H_X)$. Applying Lemma \ref{le:cutflow} to $\flow_i,\dem_i\circ v_i$, we see that the congestion of $\flow_i$ on $Cut(V(H_X))$ is at least 2, and those on $Cut(V(H_Y))$ and $Cut(V(H_Z))$ are both at least 1. Hence,  the congestion of $\flow_i$ on $Cut(W)$ is at least 4.
\end{itemize}
Since the total capacity of $Cut(W)$ is $3l$ which upper-bounds the total congestion, we get Fact 1.
%
%
\item Fact 2: the sets $v_1,\cdots,v_k$ are pairwise disjoint. 

For contradiction, suppose without loss of generality that $v_1=w_1,v_2=w_2,x_1\in w_1\bigcap w_2$. Applying Lemma \ref{le:cutflow} to multi-commodity flow $\flow_i,1\le i\le l$ and command vectors $\dem_i\circ v_i,1\le i\le l$, we have $\sum_{e\in Cut(W)} \flow_i(e)=3l$. This implies that $\flow_i(e)=1$ for any $e\in Cut(W)$. Namely, each edge in $Cut(W)$ is full of upward flow. Likewise, each edge in $Cut(T_X)$ is also full of upward flow. Let $e=\langle t_X,x_1\rangle$. Then we have $f_1(e')=f_2(e')=0$ for any edge $e'\neq e$ in $H_X$, since flow along such an edge can't reach $w_1$ or $w_2$. This, together with the precondition that $\flow_1$ satisfies $\dem_1\circ v_1$, implies $\flow_1(e)=1$. Likewise, $\flow_2(e)=1$. A contradiction is reached since $\capa(e)=1$.
\end{enumerate}
\end{proof}

\subsection{\uprob~on Trees}
Theorem \ref{thm:UndiNPC} indicates that \uprob~is hard to solve on general graphs, but does not exclude the possibility of an efficient algorithm solving \uprob~for some important special case. Indeed, \uprob~on trees allows a fast algorithm, as presented in Algorithm \ref{alg:tree}. Actually, networks with tree structure is the also the center of related literature \cite{SakashitaMinimum, Tansel1983Location,Chekuri2013The,MamadaOptimal2002,MamadaAn2006,Ito2009The,Andreev2009Simultaneous}.

Without loss of generality, trees will be arbitrarily rooted, so the concepts of ancestors, descendants, and subtrees are well defined as usual. Given vertices $u,v$ of a tree, we write $u\prec v$ if $u$ is a descendant of $v$, and $u\preceq v$ if $u\prec v$ or $u=v$.

Let's begin with a polynomial-time algorithm, which turns out to exactly solve \uprob~on trees.

\vspace{-0.4cm}
\begin{algorithm}[H]
 \textbf{Input:} a capacitated tree $G=(V,E,\capa)$, supply vectors $\dem_1,\cdots,\dem_k$\\
 \textbf{Output:} $v_i\in V,1\le i\le k$
 \begin{algorithmic}[1]
\For {each $1\le i\le k$}
        \State Let $v_i$ be the lowest common ancestor of the sources of $\dem_i$
        \While{there is a child $u$ of $v_i$ such that $\sum_{v\not\preceq u}|\dem_i(v)|<\sum_{v\preceq u}|\dem_i(v)|$} \label{line:whilecondition}
                \State Let $v_i$ be $u$
        \EndWhile
\EndFor     
\State Output($v_1,\cdots,v_k$) 
 \caption{The algorithm for \uprob~on trees.} \label{alg:tree}
 \end{algorithmic}
 \end{algorithm}
\begin{theorem}\label{thm:unditree}
The output of Algorithm \ref{alg:tree} is an optimum solution to \uprob~on trees.
\end{theorem}
\begin{proof}
Given a capacitated tree $G=(V,E,\capa)$ and supply vectors $\dem_1,\cdots,\dem_k\in \nreal^V$, let $v_1,\cdots,v_k$ be the output of Algorithm \ref{alg:tree}. Orient any edge of $G$ upward, i.e., from a vertex to its parent. The theorem is proven in two steps.

\textbf{Step 1: }Arbitrarily fix $1\le i\le k$. We claim that for any $w\in V$, any $\lambda>0$, and any flow $\flow$ satisfying $\lambda\dem_i\circ w$, there is a flow $\flow'\lesssim \flow$ which satisfies $\lambda\dem_i\circ v_i$.

The claim is proved by induction on the hop distance (i.e., the number of edges) between $v_i$ and $w$, denoted by $dist(v_i,w)$.

\textbf{Basis}: The claim trivially holds when $dist(v_i,w)=0$.

\textbf{Hypothesis}: The claim holds when $dist(v_i,w)< \delta$.

\textbf{Induction:} $dist(v_i,w)=\delta>0$. Let $x$ be the lowest common ancestor of the sources of $\dem_i$. We proceed case by case.

\textbf{Case 1: $v_i\prec w$.} 

If $x\prec w$, set flow $\flow'$ such that for any edge $e\in E$, 
$$\flow'(e)=\begin{cases}
                       \flow(e)& \textrm{if } e \textrm{ lies in the subtree rooted at }x\\
                       0& \textrm{otherwise}
                  \end{cases}.$$   
One can easily check that $\flow'\lesssim \flow$ and $\flow'$ satisfies $\lambda\dem_i\circ x$. Let $w'=x$.

If $w\preceq x$, it must happen that $v_i=w$ at the beginning of some ``while loop"  of Algorithm \ref{alg:tree} when handling $\dem_i$. That loop must assign $u$ to $v_i$, where $u$ is the child of $w$ satisfying the condition in Line \ref{line:whilecondition}. Note that $u$ lies on the path between $w$ and the final $v_i$. Set flow $\flow'$ such that for any edge $e\in E$, 
$$\flow'(e)=\begin{cases}
                       \lambda \sum_{v\not\preceq u}\dem_i(v)& \textrm{if } e=\langle u,w\rangle\\
                       \flow(e)& \textrm{otherwise}
                  \end{cases}.$$   
By Lemma \ref{le:cutflow}, we see that $\flow(\langle u,w\rangle)=\lambda\sum_{v\preceq u}|\dem_i(v)|$. Then the condition in Line \ref{line:whilecondition} implies $\flow'\lesssim \flow$. Furthermore, one can check that $\flow'$ satisfies $\lambda\dem_i\circ u$. Let $w'=u$.


\textbf{Case 2: $w\prec v_i$.} Let $y$ be the child of $v_i$ such that $w\preceq y$. Let $E_{w,v_i}$ be the edges on the path between $w$ and $v_i$. Define flow $\flow'$ such that for any edge $e\in E$, 
$$\flow'(e)=\begin{cases}
                       \lambda \sum_{v\preceq u}|\dem_i(v)|& \textrm{if } e=\langle u,u'\rangle\in E_{w,v_i}\textrm{ with }u\prec u'\\
                       \flow(e)& \textrm{otherwise}
                  \end{cases}.$$   
Since Algorithm \ref{alg:tree} outputs $v_i$ rather than $y$ for $\dem_i$, it must hold that 
\begin{eqnarray}\label{equa:conditioninline3}
\sum_{v\not\preceq y}|\dem_i(v)|\ge \sum_{v\preceq y}|\dem_i(v)|.
\end{eqnarray}
For any edge $e=\langle u,u'\rangle \in E_{w,v_i}$ with $u\prec u'$, we have 
$$\begin{array}{rll}
|\flow(e)|&=\lambda\sum_{v\not\preceq u}|\dem_i(v)|& \textrm{by Lemma \ref{le:cutflow}}\\
&\ge \lambda\sum_{v\not\preceq y}|\dem_i(v)|& \textrm{by }u\preceq y\\
&\ge \lambda\sum_{v\preceq y}|\dem_i(v)|& \textrm{by Inequality \eqref{equa:conditioninline3}}\\
&\ge \lambda\sum_{v\preceq u}|\dem_i(v)|& \textrm{by }u\preceq y\\
&=|\flow'(e)|.&
\end{array}$$

Hence, $\flow'\lesssim \flow$. One can also check that $\flow'$ satisfies $\lambda\dem_i\circ v_i$. Let $w'=v_i$.


\textbf{Case 3: neither $w\preceq v_i$ nor $v_i\preceq w$.} Let $y$ be the lowest common ancestor of $w$ and $v_i$. We have either $x\prec y$ or $y\prec x$.

If $x\prec y$, define flow $\flow'$ such that for any edge $e\in E$, 
$$\flow'(e)=\begin{cases}
                       \flow(e)& \textrm{if } e \textrm{ lies in the subtree rooted at }x\\
                       0& \textrm{otherwise}
                  \end{cases}.$$   
Then $\flow'\lesssim \flow$ and $\flow'$ satisfies $\lambda\dem_i\circ v_i$. Let $w'=x$.

If $y\prec x$, $y$ lies on the path between $v_i$ and $x$. Hence, it must happen that $v_i=y$ at the beginning of some ``while loop"  of Algorithm \ref{alg:tree} when handling $\dem_i$. 
Then that loop does not choose the subtree of $y$ containing $w$. Follow the argument of Case 2, there is a flow $\flow'\lesssim \flow$ which satisfies the demand vector $\lambda\dem_i\circ y$. Let $w'=y$.

Altogether, we always have a flow $\flow'\lesssim \flow$ which satisfies the demand vector $\lambda\dem_i\circ w'$. Because $dist(v_i,w')<dist(v_i,w)=\delta$, we apply the induction hypothesis and finish step 1.

\textbf{Step 2: }Let $\lambda^*=\lambda(G;\dem_1,\cdots,\dem_k)$. Choose $w_1,\cdots,w_k\in V$ such that there is a valid multi-commodity flow $\flow_1,\cdots,\flow_k$ satisfying $\lambda^*\dem_1\circ w_1,\cdots,\lambda^*\dem_k\circ w_k$. For any $1\le i\le k$, apply the claim in step 1 to $w_i$ and $\flow_i$, resulting in a flow $\flow'_i\lesssim \flow_i$ which satisfies $\lambda^*\dem_i\circ v_i$. Therefore, we get a valid multi-commodity flow $\flow'_1,\cdots,\flow'_k$ satisfying $\lambda^*\dem_1\circ v_1,\cdots,\lambda^*\dem_k\circ v_k$. This means that the output of Algorithm \ref{alg:tree} is an optimum solution to \uprob.
\end{proof}

\subsection{Approximation Algorithm on General Graphs}

Theorem \ref{thm:unditree} suggests that \uprob~is not extremely intractable, at least in a special case. Fortunately, the tractability can be extended to more general graphs, in the sense of approximation. Let's begin with a lemma, which shows the important role of \emph{master sources} (defined below) in approximating \uprob. 

Arbitrarily fix a supply vector $\dem$ on a capacitated graph $G=(V,E,\capa)$. Arbitrarily choose $\theta\ge |S|>1$, where $S$ is the set of sources of $\dem$. Let $w$ be a master source of $\dem$, namely $w=\argmax_{v\in V}|\dem(v)|$. 
\begin{lemma}\label{le:mainsourcegood}
For any $u\in V$ and flow $\flow$ satisfying $\dem\circ u$, there is a flow $\flow'\lesssim \flow$ which satisfies $\frac{1}{\theta-1}\dem\circ w$.
\end{lemma}
\begin{proof}  
We proceed case by case.

\textbf{Case 1: $u\notin S$.} For any $s\in S$, define demand vector $\dem_s$ such that for any $v\in V$,
$$\dem_s(v)=\begin{cases}
                      \dem(s)& \textrm{if } v=s\\
                      -\dem(s)& \textrm{if } v=u\\
                       0& \textrm{otherwise}
                  \end{cases}.$$   
By Lemma \ref{le:decomposable}, $\flow$ has a decomposition $\{\flow_s: s\in S\}$ satisfying $\{\dem_s: s\in S\}$.

Now for any $s\in S\setminus\{w\}$, define flow $\flow'_s$ such that for any $e\in E$, $$\flow'_s(e)=\frac{1}{\theta-1}\left(\flow_s(e)-\flow_w(e)\frac{\dem(s)}{\dem(w)}\right),$$ and demand vector $\dem'_s$ such that for any $v\in V$,
$$\dem'_s(v)=\begin{cases}
                      \dem(s)& \textrm{if } v=s\\
                      -\dem(s)& \textrm{if } v=w\\
                       0& \textrm{otherwise}
                  \end{cases}.$$

Our task is reduced to establishing three claims.

\textbf{Claim 1: }for any $s\in S\setminus\{w\}$, $\flow'_s$ satisfies $\frac{1}{\theta-1}\dem'_s$.

It suffices to show $\phi(v,\flow'_s)=\frac{1}{\theta-1}\dem'_s(v)$ for any $v\in V$, where $\phi(x,\vec{g})=\sum_{e\in E_-(x)}\vec{g}(e)-\sum_{e\in E_+(x)}\vec{g}(e)$, which is the net incoming of flow $\vec{g}$ at vertex $x$. Obviously, $\phi(x,\vec{g})$ is linear in $\vec{g}$.

Arbitrarily fix $v\in V$. By definition of $\flow'_s$, 
$$\begin{array}{rcl}
\phi(v,\flow'_s)&=&\frac{1}{\theta-1}\phi\left(v,\flow_s-\flow_w\frac{\dem(s)}{\dem(w)}\right)\\
&=&\frac{1}{\theta-1}\left(\phi(v,\flow_s)-\frac{\dem(s)}{\dem(w)}\phi(v,\flow_w)\right)\\
&=&\frac{1}{\theta-1}\left(\dem_s(v)-\frac{\dem(s)}{\dem(w)}\dem_w(v)\right) \qquad(\textrm{since }\flow_s,\flow_w\textrm{ satisfy }\dem_s,\dem_w)\\
&=&\frac{1}{\theta-1}\dem'_s(v) \qquad(\textrm{by definition of }\dem_s,\dem_w,\dem'_s)
\end{array}$$

\textbf{Claim 2: }$\flow'=\sum_{s\in S\setminus\{w\}}\flow'_s$ satisfies $\frac{1}{\theta-1}\dem\circ w$. It immediately follows from Claim 1.

\textbf{Claim 3: }$\flow'\lesssim \flow$.

It holds because for any $e\in E$,
$$\begin{array}{rcl}
|\flow'(e)|&=&|\sum_{s\in S\setminus\{w\}}\flow'_s(e)|\\
&\le&\sum_{s\in S\setminus\{w\}}|\flow'_s(e)|\\
&=&\sum_{s\in S\setminus\{w\}}\frac{1}{\theta-1}\left|\flow_s(e)-\flow_w(e)\frac{\dem(s)}{\dem(w)}\right|\\
&\le&\sum_{s\in S\setminus\{w\}}\frac{|\flow_s(e)|}{\theta-1}+\frac{|\flow_w(e)|}{\theta-1}\sum_{s\in S\setminus\{w\}}\frac{\dem(s)}{\dem(w)}\\
&\le&\sum_{s\in S\setminus\{w\}}|\flow_s(e)|+|\flow_w(e)|\\
&=&|\flow(e)|
\end{array}$$

The proof of Case 1 finishes.

\textbf{Case 2: $u = w\in S$.} The lemma trivially holds.

\textbf{Case 3: $u \in S\setminus \{w\}$.} 

The proof of Case 1 almost works, except that $\dem_u$ is not well-defined and the decomposition of $\flow$ does not include $\flow_u$. As a result, we still apply the proof of Case 1, after defining $\flow_u\in \real^E$ and $\dem_u\in \real^V$ to be all-zero vectors.
\end{proof} 
\begin{remark}\label{rem:etavalid}
Lemma \ref{le:mainsourcegood} remains true if $\theta$ is replaced by $\eta\ge \frac{\sum_{v\in V}\dem(v)}{\dem(w)}$.
\end{remark}

Algorithm \ref{alg:undigraph} is a simple algorithm for \uprob~with guaranteed approximation ratio. 
\begin{algorithm}[H]
 \textbf{Input:} a capacitated graph $G=(V,E,\capa)$, supply vectors $\dem_1,\cdots,\dem_k$\\
 \textbf{Output:} $w_i\in V,1\le i\le k$
 \begin{algorithmic}[1]
 \For {each $1\le i\le k$}
        \State Output $w_i=\argmax_{v\in V}|\dem_i(v)|$ as the target of $\dem_i$
 \EndFor  
\caption{An approximation algorithm for \uprob.} \label{alg:undigraph}
 \end{algorithmic}
 \end{algorithm}
  
\begin{theorem}\label{thm:undigraphapp}
Algorithm \ref{alg:undigraph} is $\max\{\theta-1,1\}$-approximate, where $\theta=\max_{1\le i\le k}|\{v\in V: \dem_i(v)<0\}|$.
\end{theorem}
\begin{proof}
Arbitrarily fix a capacitated graph $G=(V,E,\capa)$ and supply vectors $\dem_1,\cdots,\dem_k\in \nreal^V$ as input to Algorithm \ref{alg:undigraph}. Let $w_1,\cdots,w_k$ be the output. If $\theta=1$, each $w_i$ the unique source of $\dem_i$, which is trivially optimum. Hence, we assume $\theta>1$ and show that establish approximation ratio $\theta-1$.

Let $\lambda^*=\lambda(G;\dem_1,\cdots,\dem_k)$. Suppose $u_1,\cdots,u_k\in V$ is an optimum solution to \uprob. This means that there is a multi-commodity flow $\{\flow_1,\cdots,\flow_k\}$ satisfying $\{\lambda^*\dem_1\circ u_1,\cdots,\lambda^*\dem_k\circ u_k\}$. 

For any $1\le i\le k$, apply Lemma \ref{le:mainsourcegood} with $w=w_i,u=u_i,\flow=\flow_i,\dem=\lambda^*\dem_i$,  getting a flow $\flow'_i\lesssim \flow_i$ which satisfies $\frac{\lambda^*}{\theta-1}\dem_i\circ w_i$. As a result, we find a valid multi-commodity flow $\{\flow'_1,\cdots,\flow'_k\}$ satisfying $\{\frac{\lambda^*}{\theta-1}\dem_1\circ w_1,\cdots,\frac{\lambda^*}{\theta-1}\dem_i\circ w_i\}$, so $\lambda(G;\dem_1\circ w_1,\cdots,\dem_k\circ w_k)\ge \frac{\lambda^*}{\theta-1}$. The proof ends.
\end{proof}
\begin{remark}\label{rem:etavalidbound}
By applying Remark \ref{rem:etavalid} rather than Lemma \ref{le:mainsourcegood}, Theorem \ref{thm:undigraphapp} remains true if $\theta$ is replaced by $\eta=\max_{1\le i\le k}\frac{\sum_{v \in V}\dem_i(v)}{\dem_i(w_i)}$ which is not bigger than $\theta$. Hereunder, this $\eta$ will be called concentration of the supply vectors. It intuitively indicates how much demands are concentrated on sources.
\end{remark}

 Note that in Remark \ref{rem:etavalidbound}, $\eta\le 1$ if the $w_i$-entry \emph{dominates} $\dem_i$ for any $1\le i\le k$, namely $|\dem_i(w_i)|\ge \sum_{v\neq w_i}|\dem_i(v)|$.
A special such case is when every supply vector has no more than 2 sources. Then by Remark \ref{rem:etavalidbound}, we immediately have the following corollary.
\begin{corollary}\label{cor:optimumfor2sources}
When every supply vector has a dominant entry, Algorithm \ref{alg:undigraph} exactly solves \uprob.
\end{corollary}

\section{Hardness and Algorithms of \dprob}\label{sec:directedgraphs}
In this section, we adapt \uprob~to networks modeled as directed graphs. Such networks have also been studied in the network flow community and frequently appear in nowadays practice. For example, only down-streaming traffics are allowed by many data servers. 

We adopt the notation and concepts in Section \ref{sec:undigraphs} in case of no ambiguity, with three exceptions:
\begin{itemize}
\item Every edge has an inherent direction and is called an arc. An arc from vertex $u$  to vertex $v$ is denoted by $(u,v)$. We usually use $G=(V,A,\capa)$ to represent a capacitated directed $G$ with vertex set $V$, arc set $A$, and capacity vector $\capa\in \preal^A$. Accordingly,  $A_-(v)$ ($A_+(v)$, respectively) stands for the set of incoming (outgoing, respectively) arcs at vertex $v$. Likewise, define $A_-(U)$ and $A_+(U)$ for vertex subset $U\subseteq V$.
\item Any arc only allows a flow in the inherent direction, so we can naturally specify a network flow using a \emph{non-negative} vector $\flow\in \preal^A$.
\item We continue to study the problem of target location for maximizing concurrent multi-commodity flow, but in the context of directed graphs. The problem will be called \dprob~to highlight the directed model.
\end{itemize}

Note that Lemmas \ref{le:MuFinP}-\ref{le:cutflow} still hold in the context of the directed graph model. 
The following theorem indicates the strong relation between \uprob~and \dprob.
\begin{theorem}\label{thm:undi2di}
\uprob~is reducible to \dprob. 
\end{theorem}
\begin{proof} 
Arbitrarily fix an capacitated graph $G=(V,E,\capa)$ and supply vectors $\dem_1,\cdots,\dem_k\in \real^V$. We will construct a capacitated direct graph $G'=(V',A,\capa')$ and supply vectors $\dem'_1,\cdots,\dem'_k\in \real^{V'}$, and prove that the construction preserves the \emph{quality} of solutions.

\textbf{Step 1:} Construct $G'$ and the supply vectors. 

The directed graph $G'$ is obtained by replacing any edge of $G$ with the \emph{diamond} gadget as illustrated in Figure \ref{fig:unditodi}. Specifically, $V'=V\bigcup\{s_e,t_e: e\in E\}$, $A=\{(u,s_e),(v,s_e),(s_e,t_e),(t_e,u),(t_e,v): e=\langle u,v\rangle\in E\}$, and for any arc $a$ in the diamond corresponding to edge $e$, $\capa'(a)=\capa(e)$. For any $1\le i\le k$,  define $\dem'_i$ such that for any $v\in V'$, 
$$\dem'_i(v)=\begin{cases}
                      \dem_i(v)& \textrm{if } v\in V\\
                       0& \textrm{otherwise}
                  \end{cases}.$$

\begin{figure}
  \centering
  \includegraphics[width=0.65\textwidth]{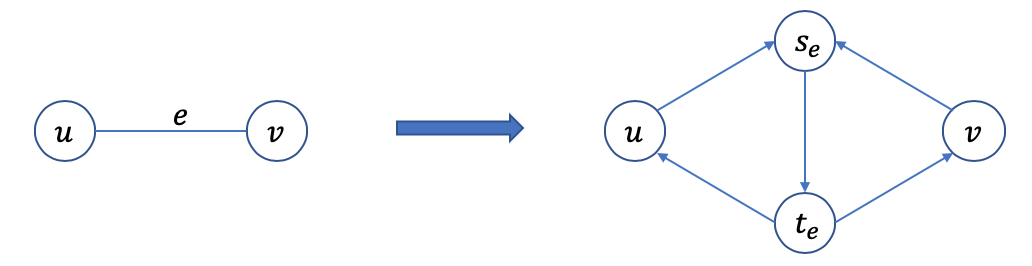}
\caption{The gadget for reducing \uprob~to \dprob.}\label{fig:unditodi}
\end{figure}   

\textbf{Step 2:} Prove that for any $v_1,\cdots,v_k\in V$, $\lambda(G;\dem_1\circ v_1,\cdots,\dem_k\circ v_k)\le \lambda(G';\dem'_1\circ v_1,\cdots,\dem'_k\circ v_k)$.

Consider any $\lambda$ and any valid multi-commodity flow $\{\flow_1,\cdots,\flow_k\}$ satisfying $\{\lambda\dem_1\circ v_1,\cdots,\lambda\dem_k\circ v_k\}$. For any $1\le i\le k$, define flow $\flow'_i$ as follows: for any $e=\langle u,v\rangle\in E$, if $\flow_i(e)$ is from $u$ to $v$, set $\flow'_i(u,s_e)=\flow'_i(s_e,t_e)=\flow'_i(t_e,v)=|\flow_i(e)|$, otherwise set $\flow'_i(v,s_e)=\flow'_i(s_e,t_e)=\flow'_i(t_e,u)=|\flow_i(e)|$; $\flow'_i(a)=0$ for any other arc $a$. It is straightforward to check that the multi-commodity flow $\{\flow'_1,\cdots,\flow'_k\}$ is valid and satisfies $\{\lambda\dem'_1\circ v_1,\cdots,\lambda\dem'_k\circ v_k\}$

\textbf{Step 3:} Prove that for any $v'_1,\cdots,v'_k\in V'$, there are $v_1,\cdots,v_k\in V$ such that $\lambda(G';\dem'_1\circ v'_1,\cdots,\dem'_k\circ v'_k)\le \lambda(G;\dem_1\circ v_1,\cdots,\dem_k\circ v_k)$.

Consider any $\lambda$ and any valid multi-commodity flow $\{\flow'_1,\cdots,\flow'_k\}$ satisfying $\{\lambda\dem'_1\circ v'_1,\cdots,\lambda\dem'_k\circ v'_k\}$. For any $1\le i\le k$, define flow $\flow_i$ as follows. For any $e=\langle u,v\rangle\in E$ oriented from $u$ to $v$, we deal case by case:
\begin{itemize}
\item When $v'_i\notin \{s_e,t_e\}$, set $\flow_i(e)=\flow'_i(u,s_e)-\flow'_i(v,s_e)$.
\item When $v'_i\in \{s_e,t_e\}$, set $\flow_i(e)=\flow'_i(u,s_e)$ if $\flow'_i(u,s_e)<\flow'_i(v,s_e)$, otherwise $\flow_i(e)=-\flow'_i(v,s_e)$.
\end{itemize}

Now for any $1\le i\le k$, we find a proper $v_i\in V$. This is also done case by case:
\begin{itemize}
\item When $v'_i\in V$, let $v_i=v'_i$.
\item When $v'_i\in \{s_e,t_e\}$ for $e=\langle u,v\rangle$, let $v_i=u$ if $\flow'_i(u,s_e)>\flow'_i(v,s_e)$, otherwise $v_i=v$.
\end{itemize}

Again, it is easy to check that the multi-commodity flow $\{\flow_1,\cdots,\flow_k\}$ is valid and satisfies $\{\lambda\dem_1\circ v_1,\cdots,\lambda\dem_k\circ v_k\}$.
\end{proof}
\begin{remark}\label{rem:DiNPC}
Theorem \ref{thm:undi2di} implies that \dprob~is at least as hard as \uprob. 
Together with Theorem \ref{thm:UndiNPC}, \dprob~is also NP-hard. 
More importantly, the reduction in the above proof preserves approximation ratio: any $\alpha$-approximation algorithm of \dprob, combined with the reduction, also $\alpha$-approximately solves \uprob.
\end{remark}
We further show that \dprob~has no PTAS.
\begin{theorem}\label{thm:DiLowerBound}
Unless P=NP, \dprob~cannot be efficiently approximated with a ratio smaller than 2.
\end{theorem}
\begin{proof}
We establish a reduction from 3-DM to \dprob~and show that the solutions to \dprob~has a big gap indicating whether or not a perfect matching exists. 

Arbitrarily fix an instance $(X,Y,Z,W)$ of 3-DM. We construct an capacitated directed graph $G=(V,A,\capa)$ and $k=|X|$ supply vectors $\dem_1,\cdots,\dem_k$. Specifically, as illustrated in Figure \ref{fig:DiLowerBound}, $G$ is adapted from the undirected graph in Figure \ref{fig:UndiNPC}, up to two modifications:
\begin{itemize}
\item The red parts, namely, vertices $t'_X,t'_Y,t'_Z$ and their incident edges, are removed.
\item All the arcs are directed upward, as indicated by the arrows.
\end{itemize}
All the arcs has capacity 1. Define supply vectors $\dem_1=\cdots=\dem_k$ such that for any $v\in V$, 
$$\dem_1(v)=\begin{cases}
                       -1& \textrm{if }v\in \{t_X,t_Y,t_Z\}\\
                       0& \textrm{otherwise}
                  \end{cases}.$$
\begin{figure}
  \centering
  \includegraphics[width=0.6\textwidth]{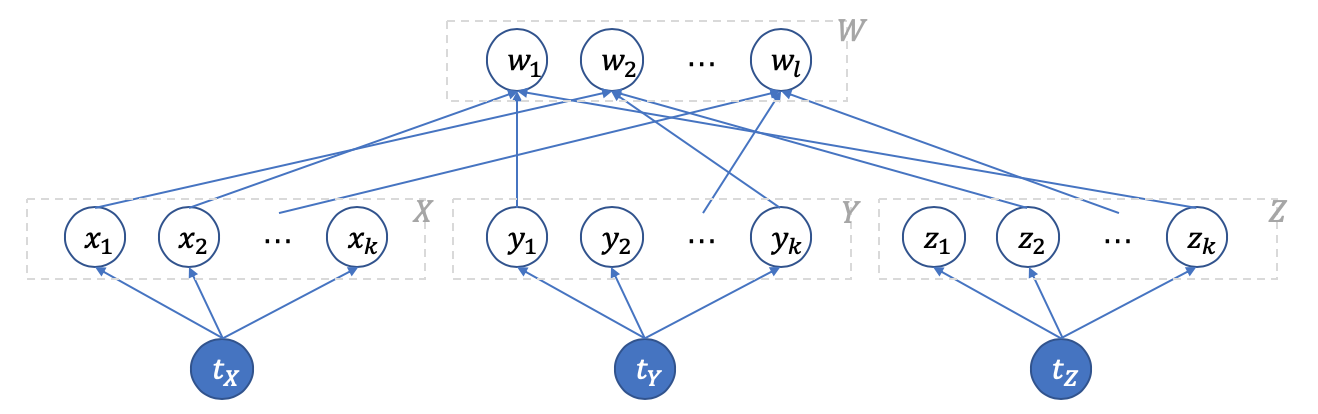}
\caption{The directed graph to which the 3-DM is reduced.}\label{fig:DiLowerBound}
\end{figure}
Our theorem immediately holds if we have the following two facts:

\textbf{Fact 1:} If $W$ contains a perfect matching, $\lambda(G;\dem_1,\cdots,\dem_k)\ge 1$.

To prove this fact, suppose without loss of generality that $\{w_1,\cdots,w_k\}$ is a perfect matching in $W$. For any $1\le i\le k$, assume $w_i=\{x,y,z\}$ with $x\in X,y\in Y,z\in Z$, and define a flow $\flow_i$ such that for any arc $a\in A$,
$$\flow_i(a)=\begin{cases}
                       1& \textrm{if }a\in \{(t_X,x),(x,w_i),(t_Y,y),(y,w_i),(t_Z,z),(z,w_i)\}\\
                       0& \textrm{otherwise}
                  \end{cases}.$$
It is straightforward to check that the multi-commodity flow $\{\flow_i:1\le i\le k\}$ is valid and satisfies $\{\dem_i\circ w_i:1\le i\le k\}$. Hence, $\lambda(G;\dem_1,\cdots,\dem_k)\ge 1$.

\textbf{Fact 2:} If $W$ contains no perfect matching, $\lambda^*=\lambda(G;\dem_1,\cdots,\dem_k)\le \frac{1}{2}$.

Let $v_1,\cdots,v_k\in V$ be such that $\lambda(G;\dem_1\circ v_1,\cdots,\dem_k\circ v_k)=\lambda^*$. One immediately sees that $v_i\in W$ for any $1\le i\le k$, unless $\lambda^*=0$. Without loss of generality, assume that $v_i=w_i$ for any $1\le i\le k$.

Let $\{\flow_i:1\le i\le k\}$ be a valid multi-commodity flow that satisfies $\{\lambda^*\dem_i\circ w_i:1\le i\le k\}$.

Since $\{w_i:1\le i\le k\}$ is not a perfect matching,  there must be $v\in X\bigcup Y\bigcup Z$ such that $|\{i:1\le i\le k, v\in w_i\}|\ge 2$. Again without loss of generality, assume that $v=x_1\in X$ and $v\in w_1\bigcap w_2$. For any $i\in \{1,2\}$, one can observe that $\flow_i(t_X,x_j)=0$ for any $2\le j\le k$, because a flow on such an arc can not reach $w_1$ or $w_2$. 

Then by Lemma \ref{le:cutflow}, $\flow_1(a)=\flow_2(a)=\lambda^*$ where $a=(t_X,x_1)$. Considering that $1=\capa(a)\ge \flow_1(a)+\flow_2(a)$, we have $\lambda^*\le \frac{1}{2}$. 
\end{proof}

To investigate the borderline of the intractability of \dprob, one might impose restrictions on instances to make them \emph{simple}. One dimension of simplification is to upper bound the source number of the supply vectors. When every supply vector has only one source, the sources altogether form a trivial optimum solution to \dprob. Hence it is reasonable to focus on bi-source supply vectors, namely those each having at most two sources. Another dimension of simplification is to focus on simple graphs, so directed trees (called di-trees) are natural candidates. A di-tree is a directed graph which, after removing the directions of the arcs and neglecting multi-edges, becomes an undirected tree. A di-path can be defined likewise. To make our result as strong as possible, we further require that the di-trees are symmetric. A capacitated directed graph is called symmetric, if (1) all arcs have equal capacity, and (2) once having an arc $(u,v)$, it also has the twin arc $(v,u)$. We will show that \dprob~remains hard even on these nearly trivial instances.

Before continuing, recall the 3-partition problem, which is well-known to be strongly NP-hard \cite[page~99]{gary1979computers}. An instance of the 3-partition problem is a multi-set $S$ of positive integers with $|S|=3m$ for some integer $m$. The objective is to decide whether $S$ has an equi-partition, namely a partition $S_1,\cdots,S_m$ of $S$ such that $\sum_{s\in S_i}s=\sum_{s\in S_j}s$ for any $1\le i,j\le m$.
\begin{theorem}\label{thm:2DiNPC}
\dprob~is NP-hard on symmetric di-paths and bi-source supply vectors
\end{theorem}
\begin{proof}
We prove the theorem via a reduction from 3-partition problems to \dprob. For this end, given an instance $S=\{s_1,\cdots,s_{3m}\}$ of 3-partition problem, we set about to construct a symmetric di-path and $(5m-2)$ bi-source supply vectors. 

\begin{figure}
  \centering
  \includegraphics[width=0.5\textwidth]{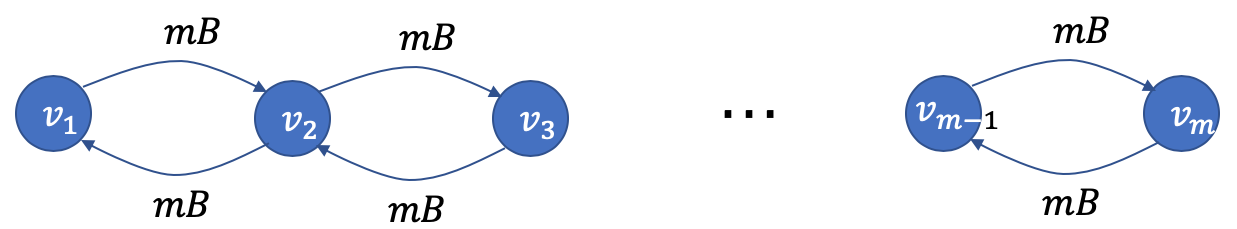}
\caption{The symmetric di-path for reducing 3-partition problem.}\label{fig:symdipathNPC}
\end{figure}
Specifically, as illustrated in Figure \ref{fig:symdipathNPC}, the di-path $G=(V,A,\capa)$ consists of $m$ vertices $v_1,\cdots,v_m$ and arcs $a_i=(v_i,v_{i+1})$ and $a'_{i+1}=(v_{i+1},v_i)$ for any $1\le i<m$. Each arc has capacity $mB$, where $B=\frac{\sum_{s\in S}s}{m}$. For any $1\le i\le 3m$ and $1\le j<m$, define supply vectors $\dem_i,\dem'_j,\dem''_j$ such that for any $v\in V$,
$$\dem_i(v)=\begin{cases}
                       -s_i& \textrm{if }v\in \{v_1,v_m\}\\
                       0& \textrm{otherwise}
                  \end{cases}$$
$$\dem'_j(v)=\begin{cases}
                       -(mB+1)& \textrm{if }v=v_j\\
                       -(m-j)B& \textrm{if }v=v_{j+1}\\
                       0& \textrm{otherwise}
                  \end{cases},$$
$$\dem''_j(v)=\begin{cases}
                       -jB& \textrm{if }v=v_j\\
                       -(mB+1)& \textrm{if }v=v_{j+1}\\
                       0& \textrm{otherwise}
                  \end{cases}.$$     
For notational simplicity, we sometimes use $\dem_{3m+1},\cdots,\dem_{5m-2}$ to stand for $\dem'_1,\cdots,\dem'_{m-1},\dem''_1,\cdots,\dem''_{m-1}$, respectively.

Our proof will be done in two steps. 

\textbf{Step 1.} If $S$ has an equi-partition, then $\lambda(G;\dem_1,\cdots,\dem_{5m-2})\ge 1$.

Let $S_1,\cdots,S_m$ be an equi-partition of $S$. For any $1\le i\le 3m$, let $1\le j\le m$ satisfy $s_i\in S_j$, and we define flow $\flow_i$ such that for any $a\in A$,
$$\flow_i(a)=\begin{cases}
                       s_i& \textrm{if }a\in\{a_k:1\le k<j\}\bigcup\{a'_k:j< k\le m\}\\
                       0& \textrm{otherwise}
                  \end{cases}.$$
One can check that $\flow_i$ satisfies demand vector $\dem_i\circ v_j$.

For any $1\le i\le m-1$, define flows $\flow'_i,\flow''_i$ such that for any $a\in A$,
$$\flow'_i(a)=\begin{cases}
                       (m-i)B& \textrm{if }a=a'_{i+1}\\
                       0& \textrm{otherwise}
                  \end{cases},$$
$$\flow''_i(a)=\begin{cases}
                       iB& \textrm{if }a=a_i\\
                       0& \textrm{otherwise}
                  \end{cases}.$$                                                       
Obviously, $\flow'_i$ satisfies $\dem'_i\circ v_i$, and $\flow''_i$ satisfies $\dem''_i\circ v_{i+1}$.

It is straightforward to check that all these flows form a valid multi-commodity flow. Altogether, we have $\lambda(G;\dem_1,\cdots,\dem_{5m-2})\ge 1$.

\textbf{Step 2.} If $\lambda(G;\dem_1,\cdots,\dem_{5m-2})\ge 1$, $S$ has an equi-partition.

Suppose $u_i,u'_j,u''_j\in V, 1\le i\le 3m, 1\le j\le m-1,$ are such that there is a valid multi-commodity flow $\flow_i,\flow'_j,\flow''_j, 1\le i\le 3m, 1\le j\le m-1,$ satisfying demand vectors $\dem_i\circ u_i,\dem'_j\circ u'_j,\dem''_j\circ u''_j, 1\le i\le 3m, 1\le j\le m-1$.

For any $1\le i\le m$, let $V_i=\{v_1\cdots,v_i\}$. We proceed in two substeps.

\textbf{Step 2.1.} For any $1\le i<m$, $u'_i=v_i$ and $u''_i=v_{i+1}$.

Arbitrary fix $1\le i<m$. For contradiction, assume that $u'_i\notin V_i$. Applying Lemma \ref{le:cutflow} to $\flow'_i,\dem'_i\circ u'_i, V_i$, one get $\flow'_i(a_i)\ge mB+1$, contradictory to the fact that $\capa(a_i)=mB$. Hence, $u'_i\in V_i$. Likewise, one can further show that $u'_i\notin V_{i-1}$. As a result, $u'_i=v_i$.

In a similar way, we also have $u''_i=v_{i+1}$.

\textbf{Step 2.2.} $S$ has an equi-partition.

Arbitrarily fix $1\le i< m$. Applying Lemma \ref{le:cutflow} to $\flow'_i,\dem'_i\circ u'_i, V_i$ and to $\flow''_i,\dem''_i\circ u''_i, V_i$ respectively, one gets 
\begin{eqnarray}\label{equa:occupycapa1}
\flow'_i(a'_{i+1})\ge (m-i)B, \flow''_i(a_i)\ge iB.
\end{eqnarray}

Let $J_i=\{j: 1\le j\le 3m, u_j\in V_i\}$. For any $j\in J_i$, apply Lemma \ref{le:cutflow} to $\flow_j,\dem_j\circ u_j, V_i$, and we have 
\begin{eqnarray}\label{equa:occupycapa2}
\flow_j(a'_{i+1})\ge s_j.
\end{eqnarray}

Likewise, for any $j\notin J_i$, applying Lemma \ref{le:cutflow} to $\flow_j,\dem_j\circ u_j, V_i$ results in  
\begin{eqnarray}\label{equa:occupycapa3}
\flow_j(a_i)\ge s_j.
\end{eqnarray}

Then, 
%
%

\begin{equation}\label{equa:occupycapa4}
\begin{split}
\begin{array}{rll}
2mB&=\sum_{1\le j\le 3m}s_j+iB+(m-i)B& \\
       &=\sum_{j\in J_i}s_j+\sum_{j\notin J_i}s_j+iB+(m-i)B& \\
       &\le \sum_{j\in J_i}\flow_j(a'_{i+1})+\sum_{j\notin J_i}\flow_j(a_i)&\\
       &\quad+\flow''_i(a_i)+\flow'_i(a'_{i+1})& \qquad \textrm{ by (\ref{equa:occupycapa1})-(\ref{equa:occupycapa3})}\\
       &\le \capa(a_i)+\capa(a'_{i+1})=2mB &\qquad \textrm{ by capacity constraints}
\end{array}
\end{split}
\end{equation}

As a result, all the inequalities in (\ref{equa:occupycapa1})-(\ref{equa:occupycapa4}) are actually equalities. Hence, $$\sum_{j\in J_i}s_j=\sum_{j\in J_i}\flow_j(a'_{i+1})=mB-\flow'_i(a'_{i+1})=iB.$$

Let $J_0=\emptyset$ and $J_m=\{j: 1\le j\le 3m\}$. For any $1\le i\le m$, define $S_i=\{s_j: j\in J_i\setminus J_{i-1}\}$, which satisfies $\sum_{s\in S_i}s=\sum_{j\in J_i\setminus J_{i-1}}s_j=\sum_{j\in J_i}s_j-\sum_{j\in J_{i-1}}s_j=B$. This means that $S_1,\cdots,S_m$ is an equi-partition of $S$.\end{proof}
\begin{remark}
Recall Corollary \ref{cor:optimumfor2sources} which implies the tractability of \uprob~on bi-source supply vectors. It is in sharp contrast to the intractability of \dprob~in this situation. Furthermore, Theorem \ref{thm:unditree} claims that \uprob~is polynomial-time solvable when the input graph is a tree, but \dprob~remains NP-hard even on symmetric di-paths. These serves as an evidence that \uprob~is generally harder than \uprob.
\end{remark}

We have seen the hardness of \dprob~even in the nearly-trivial cases. Fortunately, the next theorem will relieve us from frustration, because it indicates the possibility to approximately solve \dprob. A new definition is needed.

Given a capacitated directed graph $G=(V,A,\capa)$, for any $u,v\in V$, let $A_{\{u,v\}}=\{(u,v),(v,u)\}\bigcap A$. Define the induced graph of $G$ to be the capacitated undirected graph $G'=(V,E,\capa')$, where $E=\{\langle u,v\rangle: u,v\in V, A_{\{u,v\}}\neq \emptyset\}$, and for any $e=\langle u,v\rangle\in E$, $\capa'(e)=\sum_{a\in A_{\{u,v\}}}\capa(a)$. Intuitively, $G'$ is obtained from $G$ by neglecting the direction of the arcs and merging the capacities of twin arcs if any. 
\begin{theorem}\label{thm:symdipathAppro}
\dprob~has a polynomial-time 2-approximation algorithm on symmetric di-trees.
\end{theorem}
\begin{proof}
Arbitrarily fix a symmetric di-tree $G=(V,A,\capa)$ and supply vectors $\dem_1,\cdots,\dem_k\in \nreal^V$. Let $G'=(V,E,\capa')$ be the induced graph of $G$ and arbitrarily orient the edges. Suppose $v_1,\cdots,v_k$ be the output of Algorithm \ref{alg:tree} when the input is $(G';\dem_1,\cdots,\dem_k)$. We set about to prove that $v_1,\cdots,v_k$ is a 2-approximate solution to \dprob~on the instance $(G;\dem_1,\cdots,\dem_k)$. 

Let $\lambda^*=\lambda(G;\dem_1,\cdots,\dem_k)$ and $\lambda'^*=\lambda(G';\dem_1,\cdots,\dem_k)$. Our task is reduced to proving two claims.

\textbf{Claim 1.} $\lambda(G;\dem_1\circ v_1,\cdots,\dem_k\circ v_k)\ge \frac{\lambda'^*}{2}$.

Let $\flow'_1,\cdots,\flow'_k\in \real^E$ be a valid multi-commodity flow satisfying $\lambda'^*\dem_1\circ v_1,\cdots,\lambda'^*\dem_k\circ v_k$, where $\lambda'^*=\lambda(G';\dem_1,\cdots,\dem_k)$.

For any $1\le i\le k$, define flow $\flow_i\in \preal^A$ such that for any arc $(u,v)\in A$, 
$$\flow_i(u,v)=\begin{cases}
\frac{ |\flow'_i(e)|}{2}&\begin{array}{l}
               \textrm{if either the orientation of }e=\langle u,v\rangle \textrm{ is from }u \textrm{ to } v \textrm{ and }\flow'_i(e)>0\\
               \textrm{or the orientation is from }v \textrm{ to } u \textrm{ and }\flow'_i(e)<0
          \end{array}\\
                       0& \textrm{otherwise}
                  \end{cases}.$$
It is straightforward to check that $\flow_1,\cdots,\flow_k$ is a valid multi-commodity flow on $G$ that satisfies $\frac{\lambda'^*}{2}\dem_1\circ v_1,\cdots,\frac{\lambda'^*}{2}\dem_k\circ v_k$. Hence, Claim 1 holds.

\textbf{Claim 2.} $\lambda'^*\ge \lambda^*$.

Let $u_1,\cdots,u_k\in V$ be such that there is a valid multi-commodity flow $\flow_1,\cdots,\flow_k\in \preal^A$ on $G$ which satisfies $\lambda^*\dem_1\circ u_1,\cdots,\lambda^*\dem_k\circ u_k$. For any $1\le i\le k$, define flow $\flow'_i\real^E$ on $G'$ as follows: for any edge $e=\langle u,v\rangle\in E$, if it is oriented from $u$ to $v$ in $G'$, set $\flow'_i(e)=\flow_i(u,v)-\flow_i(v,u)$. Roughly speaking, each $\flow'_i$ is obtained from $\flow_i$ by merging traffics on twin arcs.

Again, it is easy to check that $\flow'_1,\cdots,\flow'_k$ is a valid multi-commodity flow on $G'$ that satisfies $\lambda^*\dem_1\circ u_1,\cdots,\lambda^*\dem_k\circ u_k$. This immediately leads to Claim 2.

Combining Claims 1 and 2, we have $\lambda(G;\dem_1\circ v_1,\cdots,\dem_k\circ v_k)\ge \frac{\lambda^*}{2}$, which means that $v_1,\cdots,v_k$ is a 2-approximate solution to \dprob~on the instance $(G;\dem_1,\cdots,\dem_k)$.
\end{proof}
Theorem \ref{thm:symdipathAppro} can be extended to general symmetric directed graphs. Recall the concept \emph{concentration} defined in Remark \ref{rem:etavalidbound}.
\begin{corollary}\label{cor:symmetricdialg}
\dprob~has a polynomial-time $2\cdot\max\{\eta-1,1\}$-approximation algorithm on symmetric directed graphs, where $\eta$ the concentration of the supply vectors.
\end{corollary}
\begin{proof}
We follow the proof of Theorem \ref{thm:symdipathAppro}. The only difference is that Algorithm \ref{alg:undigraph} rather than Algorithm \ref{alg:tree} is invoked. This modification is necessary, since Algorithm \ref{alg:tree} is  unfit for general undirected graphs.

The detailed proof is omitted.
\end{proof}

\section{Other Variants of \uprob}
We continue to handle other variants of \uprob, which are defined by extending \uprob~in three dimensions: 
\begin{enumerate}
\item Different network models. In Section \ref{sec:directedgraphs}, we have thoroughly studied directed and undirected graphs. This section will consider the unsplittable flow model, which means that any flow from a source to a target is along one path. Such a flow model has been actively studied in the literature \cite{Kolman2002Improved}. 
\item Different solution constraints. We restrict the targets to be chosen from a candidate set, rather than from the entire vertex set. This properly models the practical situation where applications can be deployed to prescribed servers. Such a restricted version of \uprob~is called restricted-\uprob. 
\item Different optimization goals. The network flow community typically serves three optimization goals: concurrent flow value which proportionately maximizes the flows, total flow value which maximizes the summation of all flows, and feasibility which maximize the number of feasible flows. Since concurrent flow value has been elaborated on in the previous sections, this section will investigate the latter two.
\end{enumerate}

Now we begin to present some results of the variants.

\textbf{Unsplittable flow:} A flow is unsplittable if it can be decomposed into flow paths each of which corresponds to the flow from one source to the target and the correspondence is one-to-one. By a flow path, we mean a flow which has non-zero congestion only along a path, and we say that a flow path passes an edge if the flow has non-zero congestion on the edge.

Since on trees there is a unique path connecting any two vertices, flows on trees are intrinsically unsplittable. Consequently, by Theorem \ref{thm:unditree}, even under the unsplittable flow model, \uprob~on trees is polynomial-time solvable. Actually, all the results in the previous sections remain true under the unsplittable flow model, since all the flows in the proofs are unsplittable. Moreover, stronger results can be obtained. See the following theorem as an example.

\begin{theorem}\label{thm:UnsplitNoPTAS}
Under the unsplittable flow model, \uprob~is NP-hard and cannot be approximated within ratio 2 in polynomial time.
\end{theorem}
\begin{proof}
Roughly speaking, we reduce 3-DM to \uprob, and show that the solutions to LoMuF has a big gap of unsplittable flows indicating whether or not a perfect matching exists.

Basically, we follow the proof of Theorem \ref{thm:UndiNPC}. Given an instance $(X,Y,Z,W)$ of 3-DM with $|X|=k$ and $|W|=l$, let $G=(V,E,\capa)$ be the capacitated undirected graph as constructed in the proof of Theorem \ref{thm:UndiNPC} (illustrated in Figure \ref{fig:UndiNPC}). We also adopt supply vectors $\dem_i, 1\le i\le l$ as in the proof of Theorem \ref{thm:UndiNPC}. For ease of reading, the vectors are redefined here. For any $1\le i\le k,k+1\le j\le l,v\in V$,
$$\dem_i(v)=\begin{cases}
                       -1& \textrm{if }v\in \{t_X,t_Y,t_Z\}\\
                       0& \textrm{otherwise}
                  \end{cases}$$
$$\dem_j(v)=\begin{cases}
                       -1& \textrm{if }v\in \{t'_X,t'_Y,t'_Z\}\\
                       0& \textrm{otherwise}
                  \end{cases}$$       
The rest of the proof consists of two parts.

\textbf{Part 1}: a perfect matching in $W$ implies $\lambda_{uf}(G;\dem_1,\cdots,\dem_k)\ge 1$, where the subscript $uf$ indicates that the objective value is under the unsplittable flow model.

The proof is identical to the counterpart of the proof of Theorem \ref{thm:UndiNPC}, so omitted here. 

\textbf{Part 2}:  If $W$ contains no perfect matching, $\lambda^*=\lambda_{uf}(G;\dem_1,\cdots,\dem_k)\le \frac{1}{2}$.

Suppose the optimum targets are $v_1,\cdots,v_l$, and the multi-commodity flow $\flow_1,\cdots,\flow_l$ is valid and satisfies $\dem_1\circ v_1,\cdots,\dem_l\circ v_l$. Under the unsplittable flow model, for any $1\le i\le l$, $\flow_i=\flow_{i,X}+\flow_{i,Y}+\flow_{i,Z}$ where $\flow_{i,X}$, called a summand path of $\flow_i$, is a flow path from $t_X$ (or $t'_X$ when $l>k$) to $v_i$, and likewise for $\flow_{i,Y}, \flow_{i,Z}$. 

Let $E_W$ be the set of edges incident to vertices in $W$. For $1\le i\le l$, it is easy to observe two facts:  
\begin{description}
\item [Fact 1]: If $v_i\in W$, each summand path of $\flow_i$ has non-zero congestion on at least one edge in $E_W$.
\item [Fact 2]: If $v_i\notin W$, there are two summand paths of $\flow_i$ each having non-zero congestions on at least two edges in $E_W$.
\end{description}

Now we proceed case by case.

\textbf{Case 1:} $v_i\notin W$ for some $1\le i\le l$. By Facts 1 and 2, considering that there are $3l$ edges in $E_W$ and $3l$ summand paths of all the flows, there must be an edge $e\in E_W$ shared by at least two flow paths. Since a flow path has congestion $\lambda^*$ on any edge along it and each edge has capacity 1, we see that $\lambda^*\le \frac{1}{2}$.

\textbf{Case 2:} $v_i=v_j\in W$ for some $1\le i\neq j\le l$. $\flow_i$ and $\flow_j$ altogetger have six summand paths, each of which arrives $v_i$. However, $v_i$ has only three incident edges, so at least two of the summand paths share an incident edge of $v_i$. Again, since a flow path has congestion $\lambda^*$ on any edge along it and each edge has capacity 1, we have $\lambda^*\le \frac{1}{2}$.

\textbf{Case 3:} $v_i$'s lie in $W$ and are pairwise different. Assume without loss of generality that $v_i=w_i$, for any $1\le i\le l$. Because $W$ contains no perfect matching, there exist $1\le i,j\le k$ such that $w_i\bigcap w_j\neq\emptyset$. Again without loss of generality, assume $x\in X\bigcap w_i\bigcap w_j$. 

If $\flow_{i,X}$ does not pass the edge $\langle t_X, x\rangle$, it must pass more than one edge in $E_W$ before reaching $w_i$. Following the argument of Case 1, we see that $\lambda^*\le \frac{1}{2}$. Likewise, we have $\lambda^*\le \frac{1}{2}$ if $\flow_{j,X}$ does not pass the edge $\langle t_X, x\rangle$. 

What's remaining is when both $\flow_{i,X}$ and $\flow_{j,X}$ pass the edge $\langle t_X, x\rangle$. One gets $\lambda^*\le \frac{1}{2}$ due to the capacity constraint on this edge.
\end{proof}

Then we show that restricting targets (i.e., targets can be chosen only in a candidate set of vertices) substantially affects the hardness of target location problems. Since the unrestricted version is a special case of the restricted one, all the hardness results (including the lower bounds of the approximation ratios) remain valid. In fact, restricting targets may make the problems harder, which is confirmed below. Recall Theorem \ref{thm:unditree} which claims that \uprob~on trees is polynomial-time solvable. Nevertheless, with restricted targets, \uprob~on trees even has no PTAS. 

Before going on, let's recall a property of 3-DM. Let $(X,Y,Z,W)$ be an instance of 3-DM. For any $u\in X\bigcup Y\bigcup Z$, define its covering set to be $\xi(u)=\{w\in W: u\in w\}$. It is known that $(X,Y,Z,W)$ remains NP-complete even on 3-covered instances, namely, $\max_{u\in X\bigcup Y\bigcup Z}|\xi(u)|\le 3$ \cite[page~221]{gary1979computers}. 

\begin{theorem}\label{thm:r-unditree}
\uprob~with restricted targets is NP-hard on trees and cannot be approximated within ratio $\frac{7}{6}$ in polynomial-time.
\end{theorem}
\begin{proof}
We prove the theorem by reducing 3-DM to \uprob. 

Arbitrarily fix a 3-covered instance $(X,Y,Z,W)$ of 3-DM. Let $k=|X|$ and $l=|W|$ with $W=\{w_1,\cdots,w_l\}$. We will construct an instance of \uprob~with restricted targets, including a capacitated undirected graph $G=(V,E,\capa)$, $l+2k$ supply vectors, and a \emph{candidate set} of vertices in which the targets of the supply vectors can be located.

Specifically, as illustrated in Figure \ref{fig:RTreeNPC}, $G$ is a tree consisting of a root $r$ and the set $W$ of leaves, and the capacity of each edge is 6. All the edges are oriented from leaves to the root.
\begin{figure}
  \centering
  \includegraphics[width=0.3\textwidth]{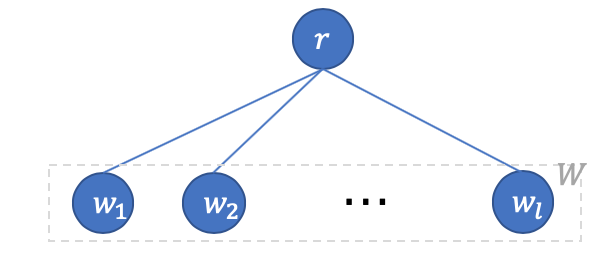}
\caption{The tree for reducing 3-DM.}\label{fig:RTreeNPC}
\end{figure}

Let $U=X\bigcup Y\bigcup Z$. For any $u\in U$, define a supply vector $\dem_u$ such that for any $v\in V$,
$$\dem_u(v)=\begin{cases}
                       -1& \textrm{if }u\in v\in W\\
                       |\xi(u)|-3& \textrm{if }v=r\\
                       0& \textrm{otherwise}
\end{cases}.$$
For any $1\le i\le l-k$, define a supply vector $\dem_i$ such that for any $v\in V$,
$$\dem_i(v)=\begin{cases}
                       -3& \textrm{if }v=r\\
                       0& \textrm{otherwise}
\end{cases}.$$
Let the candidate set be $W$, i.e., we are not allowed to choose the root as targets.

The rest of the proof consists of two parts.

\textbf{Part 1}: If $W$ contains a perfect matching, $\lambda_{W}(G;\dem_1,\cdots,\dem_{l+2k})\ge 1$, where the subscript $W$ indicates that the candidate set for the targets is $W$.

Without loss of generality, suppose $W'\subseteq W$ is a perfect matching. Let $\phi: U\rightarrow \{1,\cdots,l\}$ be the mapping such that $u\in w_{\phi(u)}\in W'$ for any $u\in U$.  For any $u\in U$, define flow $\flow_u$ such that for any edge $e=\langle r,w\rangle$,
$$\flow_u(e)=\begin{cases}
                       -2& \textrm{if }w=w_{\phi(u)}\\
                       1& \textrm{if }w\in\xi(u)\setminus\{w_{\phi(u)}\}\\
                       0& \textrm{otherwise}
\end{cases}.$$
One can check that $\flow_u$ satisfies the demand vector $\dem_u\circ w_{\phi(u)}$.

Then arbitrarily fix a bijective mapping $\psi:\{1,\cdots,l-k\}\rightarrow \{1,\cdots,l\}\setminus \phi(U)$. For any $1\le i\le l-k$, define flow $\flow_i$ such that for any edge $e=\langle r,w\rangle$,
$$\flow_i(e)=\begin{cases}
                       -3& \textrm{if }w=w_{\psi(i)}\\
                       0& \textrm{otherwise}
\end{cases}.$$
One can check that $\flow_i$ satisfies the demand vector $\dem_i\circ w_{\psi(i)}$.

Furthermore, it is easy to see that the $l+2k$ flows form a valid multi-commodity flow. Hence we finishes the proof of Part 1.

\textbf{Part 2}:  If $W$ contains no perfect matching, $\lambda^*=\lambda_{W}(G;\dem_1,\cdots,\dem_{l+2k})\le \frac{6}{7}$.

Let $v_u\in W$ for $u\in U$, $v_i\in W$ for $1\le i\le m-n$ be such that there is a valid multi-commodity flow $\flow_u$ for $u\in U$, $\flow_i$ for $1\le i\le m-n$ satisfying $\lambda^*\dem_u\circ v_u$ for $u\in U$, $\lambda^*\dem_i\circ v_i$ for $1\le i\le m-n$.

First of all, for any edge $e=\langle r,w\rangle$, we can observe two facts: 
\begin{eqnarray}\label{equa:r-unditree1}
\sum_{u\in U}|\flow_u(e)|\ge (3+n+3n')\lambda^*
\end{eqnarray}
where $n=|\{u:u\in w,v_u=w\}|$ and $n'=|\{u:u\notin w,v_u=w\}|$, and 
\begin{eqnarray}\label{equa:r-unditree2}
\sum_{1\le i\le l-k}|\flow_i(e)|\ge 3m\lambda^*
\end{eqnarray}
where $m=|\{i:1\le i\le l-k,v_i=w\}|$. The detailed proof is omitted since the inequalities are immediate results of applying Lemma \ref{le:cutflow} to $Cut(\{w\})$.

Then we proceed case by case.

\textbf{Case 1}: $v_i=v_u$ for some $1\le i\le m-n, u\in U$.

Let $e=\langle r,v_i\rangle$. By (\ref{equa:r-unditree1}) and (\ref{equa:r-unditree2}), the total congestion on edge $e$ satisfies $\sum_{u\in U}|\flow_u(e)|+\sum_{1\le i\le l-k}|\flow_i(e)|\ge 7\lambda^*$. By capacity constraint on $e$, we have $\lambda^*\le \frac{6}{7}$.

\textbf{Case 2}: $v_i=v_{j}$ for some $1\le i\neq j\le l-k$.

Let $e=\langle r,v_i\rangle$. By (\ref{equa:r-unditree1}) and (\ref{equa:r-unditree2}), the total congestion on edge $e$ satisfies $\sum_{u\in U}|\flow_u(e)|+\sum_{1\le i\le l-k}|\flow_i(e)|\ge 9\lambda^*$. By capacity constraint on $e$, we have $\lambda^*\le \frac{2}{3}$.

\textbf{Case 3}: there exists $w\in W$ such that $|\{u\in U: v_u=w\}|\ge 4$.

Let $e=\langle r,w\rangle$. By (\ref{equa:r-unditree1}), $\sum_{u\in U}|\flow_u(e)||\ge 7\lambda^*$. By capacity constraint on $e$, we have $\lambda^*\le \frac{6}{7}$.

The rest of the proof will assume that none of the Cases 1-3 happens. Let $W'=\{w\in W: v_u=w\text{ for some }u\in U\}$ and $W''=\{w\in W: v_i=w\text{ for some }1\le i\neq j\le l-k\}$. We have 

\begin{eqnarray}\label{equa:r-unditree3}
W'\bigcap W''=\emptyset, |W''|=l-k, |W'|\le k.
\end{eqnarray}

By the pigeon hole principle, one further sees that for any $w\in W'$, $|\{u\in U: v_u=w\}|=3$.

\textbf{Case 4}: there exists $u\in U$ such that $u\notin v_u$.

Let $e=\langle r,v_u\rangle$. Since $|\{u'\in U: v_{u'}=v_u\}|=3$ and $u\notin v_u$, by (\ref{equa:r-unditree1}),  $\sum_{u\in U}|\flow_u(e)||\ge 8\lambda^*$. By capacity constraint on $e$, we have $\lambda^*\le \frac{3}{4}$.

\textbf{Case 5}: None of the above cases happens.

Since $u\in v_u$ for any $u\in U$, $|U|\le |\bigcup_{w\in W'}w|$ which is at most $3k$ due to (\ref{equa:r-unditree3}). Recall that $|U|=3k$, so $|\bigcup_{w\in W'}w|=3k$. As a result, $w\bigcap w'=\emptyset$ for any $w,w'\in W'$, which implies that $W'$ is a perfect matching, contradictory to the assumption that $W$ contains no perfect matching. Therefore, Case 5 never happens.

To sum up all the cases, $\lambda^*\le \frac{6}{7}$. The proof ends.
\end{proof}

The following theorem is also a surprise. In the unrestricted case, if all the supply vectors are uni-source (i.e., each having a single source), a trivial optimum solution to \uprob~is choosing the sources themself as targets. However, when targets are restricted to a prescribed sets, \uprob~becomes NP-hard even on uni-source supply vectors and stars (i.e., trees of depth $1$).

\begin{theorem}\label{thm:unisourceRNPC}
\uprob~with restricted targets is NP-hard on uni-source supply vectors and stars.
\end{theorem}
\begin{proof}
We prove the theorem via a reduction from 3-partition problem to \uprob. For this end, given an instance $S=\{s_1,\cdots,s_{3m}\}$ of 3-partition problem, we set about to construct an instance of \uprob~with restricted targets, including a capacitated star, $3m$ supply vectors, and a \emph{candidate set} of targets.

\begin{figure}
  \centering
  \includegraphics[width=0.28\textwidth]{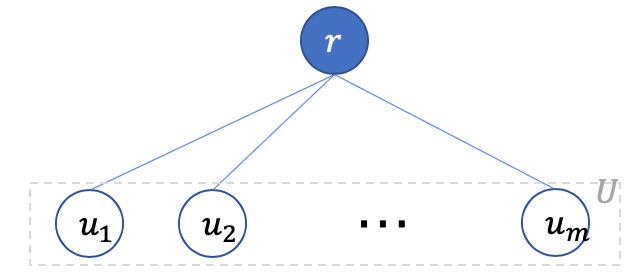}
\caption{The star for reducing 3-partition problem.}\label{fig:unisourceRNPC}
\end{figure}
Specifically, as illustrated in Figure \ref{fig:unisourceRNPC}, the capacitated undirected star $G=(V,E,\capa)$ consists of the center $r$ and the set $U$ of $m$ leaves $u_1,\cdots,u_m$. Orient every edge to point to $r$. Each edge has capacity $B$, where $B=\frac{\sum_{s\in S}s}{m}$. For any $1\le i\le 3m$, define supply vector $\dem_i\in \nreal^V$ such that $\dem_i(r)=-s_i$ and $\dem_i(u)=0$ for any $u\in U$. Appoint $U$ to be the candidate set of targets.

Our proof will be done in two steps. 

\textbf{Step 1.} If $S$ has an equi-partition, then $\lambda_U(G;\dem_1,\cdots,\dem_{3m})\ge 1$.

Let $\{S_1,\cdots,S_m\}$ be an equi-partition of $S$. For any $1\le i\le 3m$, let $1\le j\le m$ satisfy $s_i\in S_j$, and we define flow $\flow_i$ such that for any $e\in E$,
$$\flow_i(e)=\begin{cases}
                       -s_i& \textrm{if }e=\langle r,u_j\rangle\\
                       0& \textrm{otherwise}
                  \end{cases}.$$
One can check that $\flow_i$ satisfies demand vector $\dem_i\circ u_j$.

Since $\{S_1,\cdots,S_m\}$ is an equi-partition of $S$, all these flows form a valid multi-commodity flow. Hence, we have $\lambda_U(G;\dem_1,\cdots,\dem_{3m})\ge 1$.

\textbf{Step 2.} If $\lambda_U(G;\dem_1,\cdots,\dem_{3m})\ge 1$, $S$ has an equi-partition.

Let $v_1,\cdots,v_{3m}\in U$ be such that there is a valid multi-commodity flow $\{\flow_1,\cdots,\flow_{3m}\}$ satisfying $\dem_1\circ v_1,\cdots,\dem_{3m}\circ v_{3m}$. For any $1\le j\le m$, let $I_j=\{1\le i\le 3m: v_i=u_j\}$. Applying Lemma \ref{le:cutflow}, we have $|\flow_i(\langle r,u_j\rangle)|\ge s_i$ for any $i\in I_j$. Due to the capacity constraint, one gets $B\ge \sum_{i\in I_j}|\flow_i(\langle r,u_j\rangle)|$. As a result, $mB\ge \sum_{1\le j\le m}\sum_{i\in I_j}|\flow_i(\langle r,u_j\rangle)|\ge \sum_{s\in S}s=mB$, meaning that $B = \sum_{i\in I_j}|\flow_i(\langle r,u_j\rangle)|$ for any $1\le j\le m$. Hence, $\{S_j=\{s_i:i\in I_j\}: 1\le j\le m\}$ is an equi-partition of $S$.
\end{proof}

Then we discuss the target location version of the maximum multi-commodity problem. Arbitrarily fix supply vectors $\dem_i\in \nreal^V,i\in I$ on a capacitated directed/undirected graph with vertex set $V$, where $I$ is a finite index set. Roughly speaking, we are to locate targets for the supply vectors so as to maximize the total flow values. In particular, we have to find $v_i,i\in I$ to maximize $\sum_{i\in I}\lambda_i\|\dem_i\|_1$, where non-negative reals $\lambda_i$'s are such that 
\begin{enumerate}
\item For any $i\in I$,  there exists a flow $\flow_i$ satisfying the demand vector $\lambda_i\dem_i\circ v_i$, and
\item $\{\flow_i: i\in I\}$ is a valid multi-commodity flow.
\end{enumerate}
It is worth noting that all the preceding results in this paper still hold (and the proofs are also valid), except that we are not sure whether the lower bounds of approximation ratio remain true. 

Finally, we investigate the target location version of the maximum feasibility problem (\umaxfeprob~for short). Intuitively, our goal is to locate the targets so as to maximize the number of satisfiable supply vectors. Formally, given a set $S$ of demand vectors on a capacitated network $G$, its feasibility $\zeta(G;S)$ is defined to be the maximum subset of $S$ that can be simultaneously satisfied, namely, $\zeta(G;S)=\max_{S'\subseteq S, \lambda(G;S')\ge 1}|S'|$. Given supply vectors $\dem_i\in \nreal^V,i\in I$ on a capacitated directed/undirected graph $G$ with vertex set $V$, the task of \umaxfeprob~is to find $v_1\in V,i\in I$ so as to maximize $\zeta(G;\dem_i\circ v_i,i\in I)$. By abusing notation, the optimum objective value will also be denoted by $\zeta(G;\dem_i,i\in I)$. 

We will show that \umaxfeprob~is hard to approximate. The proof relies on a reduction from the well-studied maximum independent set problem (MIS) which aims to find a maximum set of vertices that are pairwise non-adjacent in a given graph. Let's first recall a property of MIS. 

\begin{lemma}[\cite{Hastad1996Clique}]\label{le:MISHard}
For any constant $\epsilon>0$, unless NP=ZPP, MIS can not be approximated within $O(n^{1-\epsilon})$ on graphs of $n$ vertices for any constant $\epsilon>0$.
\end{lemma}
\begin{theorem}\label{thm:undimaxfeasibility}
For any constant $\epsilon>0$, unless NP=ZPP, the \umaxfeprob~problem on $k$ supply vectors cannot be approximated within $O(k^{1-\epsilon})$ in polynomial-time.
\end{theorem}
\begin{proof}
We prove by reducing MIS to \umaxfeprob. Namely, given a graph $G=(V,E)$, we will construct a capacitated graph $G'=(V',E',\capa)$ and supply vectors $\dem_1,\cdots,\dem_k\in \nreal^{V'}$, where $k=|V|$.

Specifically, $V'=V\bigcup E\bigcup W$, where $W=\{w_e: e\in E\}$. $E'=E'_1\bigcup E'_2$, where $E'_1=\{\langle v,e\rangle: v\in V, e\in E, v\textrm{ is an end of }e\}$ and $E'_2=\{\langle e,w_e\rangle: e\in E\}$. Every edge of $G'$ has capacity 1. The graph $G'$ is illustrated in Figure \ref{fig:undimaxfeasibility}. We choose to orient every edge upward.
\begin{figure}
  \centering
  \includegraphics[width=0.4\textwidth]{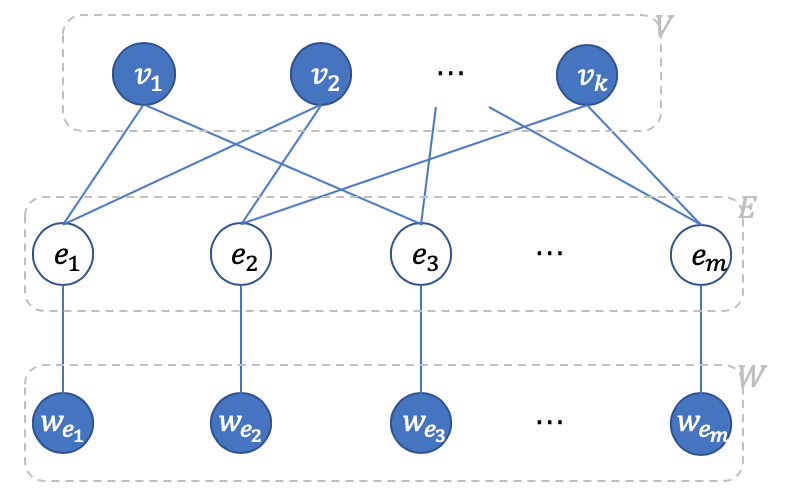}
\caption{The capacitated graph to which a MIS instance is reduced.}\label{fig:undimaxfeasibility}
\end{figure}

For any $1\le i\le k$, define a supply vector $\dem_i$ such that for any $v\in V'$,
$$\dem_i(v)=\begin{cases}
                       -1& \textrm{if }v\in \{v_i\}\bigcup\{w_e: e\textrm{ is incident to }v_i\}\\
                       0& \textrm{otherwise}
                  \end{cases}.$$

Arbitrarily fix a subset $I\subseteq \{1,\cdots,k\}$. We prove two claims:
\begin{itemize}
\item Claim 1: If $\{v_i:i\in I\}$ is an independent set of $G$, then $\lambda(G';\dem_i,i\in I)\ge 1$. 

Suppose $\{v_i:i\in I\}$ is an independent set of $G$. For any $i\in I$, define flow $\flow_i$ such that for any $e'=\langle u,e\rangle\in E'$ with $e\in E$,
$$\flow_i(e')=\begin{cases}
                       1& \textrm{if }u\in\{v_i,w_e\},e\textrm{ is incident to }v_i \textrm{ in }G\\
                       0& \textrm{otherwise}
                  \end{cases}.$$
%
%
%
It is easy to check that $\{\flow_i: i\in I\}$ is a valid multi-commodity flow satisfying $\{\dem_i\circ v_i: i\in I\}$. Hence, $\lambda(G';\dem_i,i\in I)\ge 1$.

\item Claim 2: If $\lambda(G';\dem_i:i\in I)\ge 1$, then $\{v_i:i\in I\}$ is an independent set of $G$. 

Assume $\lambda(G';\dem_i:i\in I)\ge 1$. Choose $u_i\in V',i\in I$ such that there is a valid multi-commodity flow $\{\flow_i:i\in I\}$ which satisfies $\{\dem_i\circ u_i:i\in I\}$. 

We set about to show that $\{v_1,\cdots,v_y\}$ is an independent set of $G$. For contradiction, suppose $i\neq i'\in I$ are such that $v_i$ and $v_i'$ are both incident to $e\in E$ in $G$. By Lemma \ref{le:cutflow}, no matter where $u_i$ lies, we always have $|\flow_i(\langle e,w_e\rangle)|\ge 1$. Likewise, we also have $|\flow_{i'}(\langle e,w_e\rangle)|\ge 1$. Considering that $\capa(\langle e,w_e\rangle)=1$, we reach a contradiction. Claim 2 holds.
\end{itemize}

By Claims 1 and 2, for any $\alpha$-approximate solution to the instance of \umaxfeprob, we can construct an $\alpha$-approximate solution to the instance of MIS, and vice versa. Then the theorem holds due to Lemma \ref{le:MISHard}.
%
\end{proof}

We have a trivial approximation algorithm for \umaxfeprob: Given a capacitated graph $G$ and $k$ supply vectors $\dem_i, 1\le i\le k$, by enumerating, find the first $1\le i\le k$ and $v\in V$ such that $\lambda(G;\dem_i\circ v)\ge 1$. 
This algorithm obviously has approximation ratio $k$, which is \emph{nearly optimum} due to Theorem \ref{thm:undimaxfeasibility}.

\begin{remark}
The above results (including the algorithm) about \umaxfeprob~can be extended to directed graphs. The proofs remain valid up to minor modifications, so detailed proof are omitted here.
\end{remark}


\section{Conclusion}
We formulated the target location problem for multi-commodity flows. It is a natural combination of the classic facility location problem and the multi-commodity flow problem, and extends both. It is interesting in theory and well-rooted in real-world applications. 

We mainly study the issue of maximizing concurrent flows, both on directed and undirected networks. It is interesting to see that the directed case makes the problem harder: the problem is efficiently solvable on undirected trees, but NP-hard on di-paths. Another separation is that the problem is efficiently solvable for bi-source supply vectors on undirected graphs, while it is NP-hard for such supply vectors on directed graphs. We have also made progress on algorithm design: in addition to an exact algorithm on trees, an approximation algorithm is proposed for arbitrary undirected graphs, which leads to algorithms on symmetric directed graphs.

As the first step towards this novel direction, there remain numerous open questions. Just mention a few.
\begin{enumerate}
\item Though an $\eta$-approximation algorithm exists on undirected networks, we know nothing about the lower bound of approximation ratio of the problem. Even whether a PTAS exists remains open. The directed situation is less satisfactory: except a trivial algorithm with approximation ratio $k$, no non-trivial approximation algorithm on general directed graphs is known. 

\item The variants deserve further studying. Since in many applications, targets can be chosen only from a candidate set, restricted version of our problem is of special interest. Cost minimization is an active topic in classic network flow problems. It can be easily defined in our framework, and is a rich research direction. One more variant has not yet been mentioned: This paper allows choosing just one target for each commodity, but what if more targets can be selected? 

\item Online versions of our problem are also well motivated. Recall the scenario of geo-distributed data analysis. The typical case is that the applications arrive sequentially in an online fashion, rather than all at once as we mentioned before. The online fashion poses special challenges in algorithm design. We are even not sure whether an algorithm exists with guaranteed competitive ratio.
\end{enumerate}

\section*{Acknowledgement}
We are grateful to the anonymous referees for detailed corrections and suggestions. We also thank Prof. Yungang Bao for his encouragement and support. Special thanks go to Dr. Laiping Zhao, who helped the authors formulating the problem.

\bibliographystyle{abbrv}
\bibliography{DiscretizationAndGap}


\end{document}